\pdfoutput=1
\documentclass[sigconf,10pt]{acmart}
\usepackage{graphicx}
\usepackage{balance}  
\usepackage{booktabs}
\usepackage[T1]{fontenc}
\usepackage[utf8]{inputenc}

\usepackage[shortlabels]{enumitem}
\setenumerate{1.,topsep=2pt,noitemsep,leftmargin=10pt}
\usepackage{listings}

\usepackage{cases}
\usepackage{xspace}
\usepackage{multirow}
\usepackage{color}
\usepackage{url}
\usepackage[labelfont+=bf]{caption}
\usepackage[font+=small,labelfont+=bf]{subcaption}

\usepackage{amsmath}
\usepackage{amsthm}
\usepackage{bm}
\usepackage{cleveref}

\newtheoremstyle{mytheorem}
  {.2em}
  {.2em}
  {\itshape}
  {}
  {\bfseries}
  {}
  {.5em}
  {\thmname{#1}\thmnumber{ #2}\thmnote{ (#3)}}
\theoremstyle{mytheorem}
\newtheorem{theorem}{Theorem}
\newtheorem{problem}{Problem}

\crefname{section}{Section}{Sections}
\Crefname{section}{Section}{Sections}
\crefname{appendix}{Appendix}{Appendices}
\Crefname{appendix}{Appendix}{Appendices}
\crefname{figure}{Figure}{Figures}
\Crefname{figure}{Figure}{Figures}
\crefname{equation}{Equation}{Equations}
\Crefname{equation}{Equation}{Equations}
\crefname{table}{Table}{Tables}
\Crefname{table}{Table}{Tables}
\crefname{lemma}{Lemma}{Lemmas}
\Crefname{lemma}{Lemma}{Lemmas}
\crefname{theorem}{Theorem}{Theorems}
\Crefname{theorem}{Theorem}{Theorems}
\crefname{problem}{Problem}{Problems}
\Crefname{problem}{Problem}{Problems}
\crefname{algorithm}{Algorithm}{Algorithms}
\Crefname{algorithm}{Algorithm}{Algorithms}

\usepackage{tikz}
\usetikzlibrary{calc}
\usepackage{pgfplots}
\pgfplotsset{compat=1.14}

\newcommand{\ph}[1]{\vspace{2mm} \noindent \textbf{#1}\quad}
\newcommand{\method}{QuickSel\xspace}

\let\sl\relax
\newcommand{\sl}{{\color{red} REMOVE THIS}\xspace} 

\newcommand{\dmv}{\texttt{DMV}\xspace}
\newcommand{\insta}{\texttt{Instacart}\xspace}
\newcommand{\synthetic}{\texttt{Gaussian}\xspace}

\newcommand{\ww}{\bm{w}}
\let\ss\relax
\newcommand{\ss}{\bm{s}}
\newcommand{\minus}{\scalebox{0.75}[1.0]{$-$}}

\DeclareMathOperator*{\argmin}{arg\,min}

\newcommand{\ignore}[1]{}

\newcommand{\tofix}[1]{{\color{red} #1}}

\newcommand{\yongjoo}[1]{}
\newcommand{\barzan}[1]{}

\definecolor{vintageblack}{HTML}{484043}
\definecolor{vintageyellow}{HTML}{FFC805}
\definecolor{vintagegreen}{HTML}{38A528}
\definecolor{vintageorange}{HTML}{FF4D25}
\definecolor{vintagepurple}{HTML}{AE56E2}

\setcopyright{acmlicensed}

\acmDOI{10.1145/3318464.3389727}

\acmISBN{978-1-4503-6735-6/20/06}

\copyrightyear{2020}
\acmYear{2020}
\acmBooktitle{Proceedings of the 2020 ACM SIGMOD International Conference on Management of Data (SIGMOD'20), June 14--19, 2020, Portland, OR, USA}
\acmConference[SIGMOD'20]{Proceedings of the 2020 ACM SIGMOD International Conference on Management of Data}{June 14--19, 2020}{Portland, OR, USA}
\acmArticle{4}
\acmPrice{15.00}

\begin{document}
\title[QuickSel: Quick Selectivity Learning with Mixture Models]{QuickSel: Quick Selectivity Learning \\with Mixture Models}

\author{Yongjoo Park$^{1}\mbox{*}$,\; Shucheng Zhong$^{2}\mbox{*}$,\; Barzan Mozafari$^2$}
\authornote{These authors contributed equally to this work.}

\affiliation{%
 \vspace{-1em}\institution{University of Illinois at Urbana-Champaign$^1$\hspace{1.0em} University of Michigan$^2$}
}
\email{yongjoo@illinois.edu\hspace{1.5em}\{joezhong,mozafari\}@umich.edu}



\begin{abstract}
Estimating the selectivity of a query is a key step in almost 
any cost-based query optimizer.
Most of today's databases rely on 
histograms or samples 
that are  
    periodically 
    refreshed by re-scanning the data
    as the underlying data changes.
Since frequent scans are costly,
    these statistics are often stale 
    and lead to poor selectivity estimates.
 As an alternative to scans,
\emph{query-driven histograms}
 have been proposed,
 which refine the histograms
 based on the actual selectivities 
        of the  observed queries.
 Unfortunately, these approaches are either too costly to use 
in practice---i.e., require an exponential number of buckets---or  
      quickly lose their advantage as 
    they observe more queries.
 
 In this paper, we propose a \emph{selectivity learning} framework, called \method, 
 which falls into the query-driven paradigm
 but does not use histograms.
 Instead, it builds an internal \emph{model} of
 the underlying data, which can 
 be refined significantly faster (e.g., only 
 1.9 milliseconds for 300 queries). 
 This fast refinement allows 
 \method to continuously 
  learn from \emph{each query} and 
  yield increasingly more accurate selectivity estimates over time.
Unlike query-driven histograms, \method relies on a mixture model and a new optimization algorithm for training its model.
Our extensive experiments on two real-world datasets
confirm that, given the same target accuracy, \method is 
34.0$\times$--179.4$\times$
 faster than state-of-the-art query-driven histograms, including ISOMER and STHoles. Further, given the same space budget, \method is 26.8\%--91.8\% more accurate than 
periodically-updated histograms and samples, respectively.

\end{abstract}

\ignore{
Estimating query answers' cardinalities is an essential task for finding an optimal query plan; for modern big data analysis, even small relative improvements in query plans can lead to large savings in absolute time. The most common cardinality estimation techniques are based on precomputed data synopses, such as histograms or samples. In this work, we show we can produce much more precise cardinality estimates by learning a data distribution from thousands of (freely available) past query answers. This result is possible from our new approach that constructs a \emph{maximum entropy distribution} consistent with past query answers with only $O(n^2)$ space; this is a significant improvement over previous maximum entropy-based methods, which require $O(2^n)$ space. This saving in space complexity also leads to a quicker data distribution learning. Thanks to this computational efficiency, we can use a much larger number of past queries (within a reasonable time-bound); thus, being able to learn a precise data distribution. Our empirical study suggests that our method is not only superior to previous learning-based cardinality estimation methods, but also superior to conventional cardinality estimation techniques, such as histograms and sampling. Also, we demonstrate that a database system using our method gradually improves the quality of its query plans as processing more queries. To the best of our knowledge, we are the first that demonstrates this database system's evolutionary behavior resulting from learning-based cardinality estimations.
}

\maketitle

\ignore{the state-of-the-art technique
     on adaptive histograms relies on the \emph{maximum entropy principle}
     to refine the bucket frequencies based on the observed selectivities 
        of the past queries.}


\section{Introduction}
\label{sec:intro}

Estimating the \emph{selectivity} of a query---the fraction of input tuples that satisfy the query's predicate---is a fundamental component in cost-based query optimization, 
including  both traditional RDBMSs~\cite{oracle_cbo,postgres_cbo,ibm_cbo,sqlserver_cbo,mariadb_cbo} and modern SQL-on-Hadoop engines~\cite{hive_cbo,spark_cbo}.
The estimated selectivities allow the query optimizer to choose the cheapest access path
or query plan~\cite{selinger1979access,kester2017access}.

Today's databases typically rely on histograms~\cite{postgres_cbo,ibm_cbo,sqlserver_cbo} or samples~\cite{oracle_cbo} for their selectivity estimation.
These structures need to be populated in advance by performing costly table scans. 
However, as the underlying data changes, 
    they quickly become stale and highly inaccurate.
This is why they need to be updated periodically, creating additional costly operations
    in the database engine
    (e.g., \texttt{ANALYZE table}).\footnote{Some database 
    systems~\cite{sqlserver_cbo} automatically update their statistics when the number of modified tuples exceeds a threshold.}

\ignore{
\begin{table*}[t]
\caption{The selectivity estimation methods that can be refined from observed queries.}
\label{tab:diff}
\centering
\small
\begin{tabular}{l l c c}
\toprule
\textbf{Method} 
  & \textbf{Representation} 
  & \textbf{High Accuracy}
  & \textbf{Fast Refinement} \\
\midrule
Freq-split Histograms~\cite{aboulnaga1999self,lim2003sash}
  & histograms & No & Yes \\[5pt]
MaxEnt Histograms~\cite{srivastava2006isomer,markl2005consistently,kaushik2009consistent,re2012understanding,markl2007consistent}
  & histograms & Yes & No \\[5pt]
\method (Ours) & mixture models & Yes & 
\begin{tabular}{@{}c@{}}\textbf{Yes!} \\ \emph{(our core contribution)}\end{tabular}
\\
\bottomrule
\end{tabular}
\end{table*}
}

\begin{table*}[t]
\caption{The differences between query-driven histograms~\cite{srivastava2006isomer,markl2005consistently,kaushik2009consistent,re2012understanding,markl2007consistent} and our method (\method)}.
\label{tab:diff}

\vspace{-4mm}

\centering
\small
\begin{tabular}{l p{40mm} p{44mm} p{66mm}}
\toprule
 & \textbf{Query-driven Histograms} & \textbf{\method (ours)} & \textbf{Our Contribution} \\
\midrule
Model  & histograms \newline (non-overlapping buckets)
& mixture models \newline (overlapping subpopulations)
& Employs a new expressive model \newline $\rightarrow$ \textbf{no exponential growth of complexity} \\
\midrule
Training  & \emph{maximum entropy} \newline solved by \newline \emph{iterative scaling}
          & \emph{min difference from a uniform \newline distribution
          } \newline solved \emph{analytically}
          & A new optimization objective and its reduction \newline to quadratic programming (solved analytically)
          \newline $\rightarrow$ \textbf{fast training and model refinements} \\
\bottomrule
\end{tabular}
\end{table*}


To address the shortcoming of  scan-based 
approaches, 
    numerous proposals for query-driven histograms have been introduced,
        which
continuously correct and refine the histograms based on the actual 
selectivities observed after running each  query~\cite{srivastava2006isomer,bruno2001stholes,markl2005consistently,kaushik2009consistent,aboulnaga1999self,lim2003sash,re2012understanding,stillger2001leo,agrawal2006autoadmin}.
 There are two approaches to query-driven histograms.
The first approach~\cite{bruno2001stholes,aboulnaga1999self,lim2003sash,agrawal2006autoadmin},
which we call \emph{error-feedback histograms},
recursively splits existing buckets (both boundaries and frequencies) for every distinct query 
observed, such that  their error
is  minimized 
for the latest query.
Since the error-feedback histograms do not minimize the (average) error across multiple queries, their estimates tend to be much less accurate.

To achieve a higher accuracy, the second approach is to compute the bucket frequencies based on the maximum entropy principle~\cite{srivastava2006isomer,markl2005consistently,re2012understanding,kaushik2009consistent}.
 However, this approach (which is also the state-of-the-art) requires solving an  optimization problem,
 which quickly becomes prohibitive  
as the number of observed queries (and hence, number of buckets) grows. 
Unfortunately, one cannot simply prune
    the buckets in this approach, 
as it will
    break the underlying assumptions of their optimization algorithm (called \emph{iterative scaling},
 see \cref{sec:prelim:drawback} for details).
Therefore, they  prune the \emph{observed queries}
instead in order to keep the optimization overhead feasible in practice. 
 However, this also means discarding data that could be used for learning a more accurate distribution.

\ph{Our Goal}
We aim to develop a new framework for selectivity estimation  that can quickly refine its \emph{model} 
    after observing each query, thereby
        producing increasingly more accurate
        estimates over time.
We call this new framework \emph{selectivity learning}.
We particularly focus on designing 
    a low-overhead method that can 
        scale to a large number of observed queries
        without requiring 
        an exponential number of buckets. 
        

\ph{Our Model}
To overcome the limitations of query-driven histograms, 
we use a \emph{mixture model}~\cite{ml-textbook} to capture the  unknown distribution of the data. 
A mixture model is a probabilistic model 
to approximate an arbitrary probability density function (pdf), say $f(x)$, using a combination of simpler pdfs:
\begin{align}
f(x) = \sum_{z=1}^m \; h(z) \; g_z(x)
\label{eq:intro:model}
\end{align}
where 
$g_z(x)$ is the $z$-th simpler pdf and 
$h(z)$ is its corresponding weight.
The subset of the data 
that follows each of the simpler pdfs is called a \emph{subpopulation}. 
Since the subpopulations 
are allowed to  
overlap with one another, a mixture model is strictly more expressive
  than histograms.
In fact, it is shown that mixture models can achieve a higher accuracy than histograms~\cite{kde_vs_hist}, which is confirmed by  our empirical study 
(\cref{sec:exp:model}).
  To the best of our knowledge, 
    we are the first to propose 
        a mixture model for selectivity estimation.\footnote{Our mixture model   also differs from \emph{kernel density estimation}~\cite{heimel2015self,gunopulos2005selectivity,blohsfeld1999comparison}, which scans the actual data, rather than analyzing observed queries.}

\ph{Challenges}
Using a mixture model for selectivity learning requires finding optimal parameter values for $h(z)$ and $g_z(x)$;
 however, this optimization (a.k.a. training) is challenging for two reasons.

First, the training process aims to find  parameters that maximize
the model \emph{quality}, defined as  
$\int Q(f(x)) \, dx$ for some  metric of quality $Q$  (e.g., entropy).
However, computing this integral is non-trivial
for a mixture model since its subpopulations may overlap in  arbitrary ways.
That is, the combinations of $m$ subpopulations can
create $2^m$ distinct 
ranges,
each with a potentially different value of $f(x)$.
As a result, na\"ively computing the quality of a mixture model quickly becomes  intractable as the number of observed queries grows.

Second, the outer  optimization algorithms
 are often iterative (e.g., iterative scaling, gradient descent), which means they have to repeatedly evaluate the model quality as they search for optimal parameter values. 
 Thus, even when 
 the model quality 
 can be evaluated relatively
 efficiently,
    the overall training/optimization process can be quite costly. 

\ph{Our Approach}
First, to ensure the efficiency of the model quality evaluation,
we propose a new optimization objective.
Specifically, we find the parameter values that minimize
the $L2$ distance (or equivalently, \emph{mean squared error}) 
between the mixture model
 and a uniform distribution, rather than maximizing the entropy of the model (as pursued by previous work~\cite{srivastava2006isomer,markl2005consistently,kaushik2009consistent,re2012understanding,markl2007consistent}).
As described above, directly computing the quality of a mixture model involves costly integrations over $2^m$ distinct ranges. 
However, when minimizing the $L2$ distance, 
the $2^m$ integrals can be 
 reduced to 
 only $m^2$ multiplications, hence greatly reducing the complexity of the model quality evaluation.
 Although minimizing the $L2$ distance is much more efficient than maximizing the entropy, these
two objectives are closely related
(see our report~\cite{anonymous_report} for a discussion).
 
In addition,
we adopt a non-conventional variant of mixture models, 
    called a \emph{uniform mixture model}.
While uniform mixture models 
    have been previously explored in limited settings (with only a few subpopulations)~\cite{craigmile1997parameter,gupta1978uniform}, 
    we find that they are quite appropriate  
     in our context as they allow for
      efficient computations of the $L2$ distance. 
 That is,
with this choice, we can 
evaluate the quality of a model
by only using  min, max, and multiplication operations (\cref{sec:method:est}).
Finally, our optimization  can be expressed as a standard 
\emph{quadratic program}, 
which still requires an iterative procedure.



Therefore, to avoid the costly iterative optimization,
we also devise an analytic solution that can be computed more efficiently.
Specifically,
in addition to the standard reduction 
(i.e., moving some of the
original  constraints to
the objective clause as penalty terms),
we completely relax the positivity constraints for $f(x)$, exploiting the fact that they will be naturally satisfied in the process of approximating the true distribution of the data.
With these modifications, 
we can solve for the solution
analytically by setting
the gradient of the objective function to zero.
This simple transformation speeds up 
the training by 1.5$\times$--17.2$\times$.
In addition, since our analytic solution requires a constant number of operations, the training time is also 
consistent across different 
datasets and workloads.




Using these ideas, we have developed a first prototype of our selectivity learning proposal, called \method,
        which allows for extremely fast model refinements.
As summarized in \cref{tab:diff}, \method  differs from---and considerably improves upon---query-driven histograms~\cite{aboulnaga1999self,lim2003sash,srivastava2006isomer,markl2005consistently,kaushik2009consistent,re2012understanding}
in terms of both modeling and training
(see \cref{sec:related} for a detailed comparison). 


\ignore{
\barzan{the following para is out of place. either move and merge it with where u discussed dynamic histograms or make it a footnote to where we say we re the first to do Mixture model}
Finally, we are aware of the selectivity estimation methods~\cite{heimel2015self,gunopulos2005selectivity,blohsfeld1999comparison} that rely on a technique called \emph{kernel density estimation} (KDE). Although KDE itself can be regarded as a special case of a mixture model, their objectives are creating \tofix{models after database scans; thus, they are fundamentally different from \method or from  selectivity learning in general}
\barzan{1) not english and completely incomprehensible.
2) what is db scans?? 3) if they re doing sel estimation then
they have the same objective anyways so makes no sense what u re saying} (see \cref{sec:related}).
}

\ph{Contributions} We make the following contributions:
\begin{enumerate}
\item We propose the first 
mixture model-based approach to
 selectivity estimation
(\cref{sec:model}).
\item 
For training the mixture model, we design a constrained optimization problem
(\cref{sec:training:prob}). 
\item We show that the proposed optimization problem 
can be reduced (from an exponentially complex one) 
to a quadratic program and present
further optimization strategies for solving it (\cref{sec:method:solving}).
\item We conduct extensive experiments 
 on two real-world datasets to 
 compare \method's performance and state-of-the-art selectivity estimation techniques  (\cref{sec:exp}).
\end{enumerate}

\ignore{
Data analysis with declarative languages is an important feature for both conventional RDMBS~\cite{oracle_online,postgres_online,stillger2001leo} and modern distributed data analytics systems~\cite{armbrust2015spark,meijer2006linq,facebook-presto}. For those systems, accurate cardinality estimations are crucial for finding an optimal query plan.
A study shows the errors in cardinality estimations caused more than 10$\times$ slowdowns for 8\% of benchmark queries and more than 100$\times$ slowdowns for 5\% of benchmark queries, for both open source and commercial database systems~\cite{leis2015good}. These slowdowns are more significant for large-scale data analysis; even a small relative efficiency difference can lead to a big absolute time difference.

In reality, accurate cardinality estimations are hard when (i) correlations exist among attributes or (ii) a query selectivity is low. These conditions make commonly used static cardinality estimation techniques (such as per-attribute histograms~\cite{jagadish2001global} or sampling based methods~\cite{leiscardinality,lipton1990practical}) inaccurate. In general, high-accuracy estimations demand more histogram bins or more sampled tuples; however, larger numbers of bins and sampled tuples inherently cause larger memory footprints and increased estimation times. Given fixed budgets of space and time, static cardinality estimation techniques are essentially suboptimal unless query workloads are completely uniform over the entire data, or future query workloads are fully known in advance (so, those techniques are optimized for those workloads). We do not make such strong assumptions.

\ph{Consistent max-entropy distribution}
In this work, we focus on learning-based cardinality estimation (LCE), a class of cardinality estimation techniques that dynamically adjust their performance according to query workloads. For LCE, it is effective to model an underlying data distribution using a maximum entropy distribution consistent with past query answers. The notion of maximum entropy enables us to select the most likely probability distribution given past query answers~\cite{srivastava2006isomer,markl2005consistently,kaushik2009consistent}. This approach is principled~\cite{jaynes1957information} and more accurate compared to other existing approaches~\cite{aboulnaga1999self,lim2003sash}.

Finding a maximum entropy distribution entails solving an optimization problem. Previous work used a dynamically-constructed histogram for representing an underlying data distribution, and the weights of the histogram bins were optimized by maximizing the entropy of the distribution. The major advantage of using a histogram is the simplicity it brings to the optimization process: the objective function of the optimization problem includes the terms that directly correspond to the histogram bins (see \cref{sec:fail:partition}). This is possible because the bins of a histogram cover disjoint regions (of a cross product of the domains of attributes); that is, there are zero intersections among the bins. The correspondence between the histogram bins and the terms in an objective function also implies that the objective function is only as complicated as the number of the histogram bins.

Since those histogram-based LCE methods represent an underlying data distribution using a histogram obtained by solving an optimization problem, the success of those methods is contingent on the quality of a preprocessing step that generates a bounded number of histogram bins (for efficiency) in a workload-sensitive manner (for high accuracy). However, this process is surprisingly hard. Several approaches we discuss in \cref{sec:voronoi} fail either because bins overlap or because those methods are not effective for the datasets with many attributes. Existing partitioning methods~\cite{srivastava2006isomer,bruno2001stholes} generate $O(2^n)$ bins, where $n$ is the number of past queries; thus, they make the optimization process prohibitively expensive. A graphical model~\cite{lim2003sash} is only straightforward for the datasets with discrete attributes. Note that we are not against developing more sophisticated techniques for workload-sensitive bin generations; however, our proposed approach obviates such efforts.

\ph{UMM and Contributions}
Our work introduces a uniform mixture model (UMM) in place of the histogram for representing an underlying data distribution. With the UMM, representing an underlying data distribution in a workload-sensitive manner becomes straightforward (\cref{sec:method:kernel}). The UMM is a variant of a Gaussian mixture model~\cite{ml-textbook}; UMM's probability density is computed as a summation of the probability densities of \emph{subpopulations}---weighted uniform distributions of different sizes at different locations. This choice of using uniform distributions for the subpopulations make the computations efficient. Also, UMM does not require ``zero intersection'' in contrast to histogram bins. This high flexibility makes the optimization problem exponentially expensive though: its objective function involves $O(2^n)$ terms, where $n$ is the number of past queries. However, we show we can reduce those terms to only $O(n^2)$ terms using a mathematical trick.

\tikzset{
    process box/.style={
        align=center,draw=black,minimum height=7mm,minimum width=10mm,font=\small,
    },
    process line/.style={
        draw=black,thick,->,
    }
}

\begin{figure}
    \centering
    
    \begin{subfigure}[b]{0.4\columnwidth}
    \centering
        \begin{tikzpicture}
        \def\r{15mm};
        \node[process box] (a1) at (0,0) {Solve\\optimization\\ problem};
        \node[process box] (a2) at (0,-2*\r) {Write\\optimization\\ problem};
        \node[process box] (a3) at (0,-3*\r) {Histogram\\ (Generate bins)};
        \node[process box,draw=none] (a4) at (0,-4*\r+2mm) {\textbf{query answers}};
        
        \draw[process line] (a4)--(a3) node[midway,right,font=\small] {$O(n)$};
        \draw[process line] (a3)--(a2) node[midway,right,font=\small] {$O(2^n)$};;
        \draw[process line] (a2)--(a1) node[midway,right,font=\small] {$O(2^n)$};;
        \end{tikzpicture}
        \caption{previous work}
    \end{subfigure}
    ~
    \begin{subfigure}[b]{0.4\columnwidth}
    \centering
        \begin{tikzpicture}
        \def\r{15mm};
        \def\y{40mm};
        \node[process box] (b1) at (\y,0) {Solve\\optimization\\ problem};
        \node[process box] (b2) at (\y,-1*\r) {Reduce\\analytically};
        \node[process box] (b3) at (\y,-2*\r) {Write\\optimization\\ problem};
        \node[process box] (b4) at (\y,-3*\r) {UMM};
        \node[process box,draw=none] (b5) at (\y,-4*\r+2mm) {\textbf{query answers}};
        
        \draw[process line] (b5)--(b4) node[midway,right,font=\small] {$O(n)$};
        \draw[process line] (b4)--(b3) node[midway,right,font=\small] {$O(n)$};
        \draw[process line] (b3)--(b2) node[midway,right,font=\small] {$O(2^n)$};
        \draw[process line] (b2)--(b1) node[midway,right,font=\small] {$O(n^2)$};
        \end{tikzpicture}
        \caption{\method}
    \end{subfigure}
    \caption{\method in the perspective of query optimization. Using a flexible model (UMM) causes the optimization problem construction to generate $O(2^n)$ terms. However, we reduce this complexity to $O(n^2)$.}
    \label{fig:intro:opt}
\end{figure}

In the query optimization perspective, our approach addresses the $O(2^n)$ complexity of existing partitioning methods (for histogram bins) by transferring the complexity of the bin generation to the optimization problem writing. We depict this concept in \cref{fig:intro:opt}. Due to this ``complexity transfer'', the optimization problem writing now generates $O(2^n)$ terms. However, we present a technique that reduces those $O(2^n)$ terms to $O(n^2)$ terms using analytic approximations. This process is possible because there is no materialization (in contrast to the bin generation process). After all, it suffices for a LCE component to compute those $O(n^2)$ terms. This computational efficiency enables our approach to build a precise data distribution by exploiting more query answers compared to previous work. From now on, we call this proposed LCE algorithm---a combination of the UMM and analytic objective function simplifications---\method.

\tofix{I should talk about actual error reductions or speedup compared to existing work. I defer those descriptions until experiments are done.}

The contribution of this work is as follows:
\begin{enumerate}
\item We introduce the uniform mixture model and show that we can obtain a maximum entropy distribution consistent with past query answers with $O(n^2)$ space costs.
\item We show empirically that \method is superior to previous LCE work and other static cardinality estimation methods.
\item We demonstrate that a database system with \method gradually improves the quality of its query plans as more queries are processed.
\end{enumerate}
}


\section{Preliminaries}
\label{sec:prelim}

In this section, we first define relevant notations in \cref{sec:prelim:notations} and then formally define the problem of \emph{query-driven selectivity estimation} in \cref{sec:prelim:prob}. 
Next, in \cref{sec:prelim:drawback}, we discuss the drawbacks of   previous approaches.

\subsection{Notations}
\label{sec:prelim:notations}

\begin{table}[t]
\caption{Notations.}
\label{tab:notations}

\vspace{-4mm}

\centering
\small
\begin{tabular}{l p{65mm}}
\toprule
Symbol & Meaning \\
\midrule
$T$                  & a table (or a relation) \\
$C_i$                & $i$-th column (or an attribute) of $T$; $i = 1, \ldots, d$ \\
$|T|$                & the number of tuples in $T$ \\
$[l_i, u_i]$         & the range of the values in $C_i$ \\
$x$                  & a tuple of $T$ \\
$B_0$                & the domain of $x$; $[l_1, u_1] \times \cdots \times [l_d, u_d]$ \\
$P_i$                & $i$-th predicate \\
$B_i$                & hyperrectangle range for the $i$-th predicate \\
$|B_i|$              & the size (of the area) of $B_i$ \\
$x \in B_i$          & $x$ belongs to $B_i$; thus, satisfies $P_i$ \\
$I(\cdot)$           & indicator function that returns 1 if its argument is true and 0 otherwise \\
$s_i$                & the selectivity of $P_i$ for $T$ \\
$(P_i, s_i)$         & $i$-th observed query \\
$n$                  & the total number of observed queries \\
$f(x)$               & probability density function of tuple $x$ (of $T$) \\
\bottomrule
\end{tabular}
\end{table}

 \cref{tab:notations} summarizes 
 the notations used throughout this paper.

\ph{Set Notations} $T$ is a relation that consists of $d$ real-valued columns $C_1, \ldots, C_d$.\footnote{Handling integer and categorical columns is discussed in \cref{sec:prelim:prob}.}
The range of values in $C_i$ is $[l_i, u_i]$ and 
the cardinality (i.e., row count) of $T$ is
$N$=$|T|$. The tuples in $T$ are denoted by $x_1, \ldots, x_N$,  where each $x_i$ is a size-$d$ vector that belongs to 
$B_0 = [l_1, u_1] \times \cdots \times [l_d, u_d]$. 
Geometrically, $B_0$ is the area bounded by a hyperrectangle whose bottom-left corner is $(l_1, \ldots, l_d)$ and top-right corner is $(u_1, \ldots, u_d)$. The size of $B_0$ can thus be computed as $|B_0|$=$(u_1 - l_1) \times \cdots \times (u_d - l_d)$.

\ph{Predicates}  We use $P_i$ to denote the (selection) predicate   of the $i$-th query on $T$. In this paper, a predicate is a 
conjunction\footnote{See \cref{sec:types} for a discussion of disjunctions and negations.}
of one or more \emph{constraints}.
Each constraint is a \emph{range constraint}, which can be 
one-sided (e.g., $3\leq C1$) or two-sided (e.g.,  $-3\leq C1\leq 10$). 
This range can be extended to also handle equality constraints on categorical data (see \cref{sec:types}).
 Each predicate $P_i$ is represented by a hyperrectangle $B_i$. For example, a constraint ``$1 \le C1 \le 3$ AND $2 \le C2$'' is represented by a hyperrectangle $(1, 3) \times (2, u2)$, where 
 $u2$ is the upper-bound of $C2$.
We use $P_o$ to denote an empty predicate, i.e., one that selects all tuples.


\ph{Selectivity}
The \emph{selectivity} $s_i$ of $P_i$ is defined as the \emph{fraction} of the rows of $T$ that satisfy the predicate. That is, $s_i = (1/N) \sum_{k=1}^N I(x_k \in B_i)$,
where $I(\cdot)$ is the indicator function. 
A pair $(P_i, s_i)$ is referred to as an \emph{observed query}.\footnote{This pair is also 
referred to as an \emph{assertion} by
prior work~\cite{re2010understanding}.} Without loss of generality, we assume that $n$  queries have been observed for $T$ and seek to estimate $s_{n+1}$.  
Finally, we use $f(x)$ to denote the joint probability density function of tuple $x$ (that has generated tuples of $T$).




\subsection{Problem Statement}
\label{sec:prelim:prob}
\label{sec:types}

Next, we formally state the problem:


\begin{problem}[Query-driven Selectivity Estimation]
Consider a set of $n$ observed queries ${(P_1, s_1), \allowbreak \ldots, \allowbreak (P_n, s_n)}$ for $T$.
By definition, we have the following for each $i=1,\ldots,n$:
\[
\int_{x \in B_i} f(x) \; dx = s_i
\]
Then, our goal is to build a model of $f(x)$ that can  estimate the selectivity $s_{n+1}$ of a new predicate $P_{n+1}$.
\label{prob:sec}
\end{problem}

\vspace{2mm}

Initially, before any query is observed,
    we can conceptually consider a default query $(P_0, 1)$, 
    where all tuples are selected and hence, the 
      selectivity is 1 (i.e., no predicates).


\ph{Discrete and Categorical Values}
\cref{prob:sec} can be extended 
to support discrete attributes (e.g., integers, characters,
categorical values) and  equality constraints on them, as follows. Without loss of generality, suppose that $C_i$ contains the integers in $\{ 1, 2, \ldots, b_i \}$. 
 To apply the solution to \cref{prob:sec}, it suffices to (conceptually) 
 treat these integers as real values in $[1, b_i + 1]$ and then convert the original constraints on the integer values into range constraints, as follows.
A constraint of the form ``$C_i = k$'' 
will be converted to a range constraint of the form 
$k \le C_i < k+1$. 
 Mathematically, this is equivalent to replacing a probability mass function with a probability density function defined using
\emph{dirac delta functions}.\footnote{A dirac delta function $\delta(x)$ outputs $\infty$ if $x = 0$ and outputs 0 otherwise while satisfying $\int \delta(x)\, dx = 1$.} 
Then, the summation of the original probability mass function can   be converted to the integration of the probability density function.
String data types (e.g., char, varchar) and their equality constraints can be similarly supported, by conceptually mapping each
string into an integer (preserving their order) and applying the {conversion described above for the integer data type. 

\ph{Supported Queries}
Similar to prior work~\cite{bruno2001stholes,aboulnaga1999self,srivastava2006isomer,markl2005consistently,kaushik2009consistent,re2012understanding,markl2007consistent,lim2003sash}, we support  selectivity estimation for 
predicates with conjunctions, negations, and disjunctions of   range and equality constraints on numeric and categorical columns.
We currently do not support  wildcard constraints (e.g., \texttt{LIKE '*word*'}), \texttt{EXISTS} constraints, or \texttt{ANY} constraints.
 In practice, often a \emph{fixed selectivity} is used for unsupported predicates,
e.g., 3.125\% in Oracle~\cite{oracle_cbo}.

To simplify our presentation, we focus on conjunctive predicates.
However, negations and disjunctions can also be easily supported.
This is because our algorithm only requires the ability to 
compute the intersection size of pairs of predicates $P_i$ and $P_j$,
    which can be done
    by converting 
    $P_i \land P_j$ into  a disjunctive normal form
    and then using  the inclusion-exclusion principle
            to compute its size.

As in the previous work, 
we focus our presentation on predicates on a single relation.
However, any selectivity estimation technique for a single relation  can be applied to estimating selectivity of a join query whenever the predicates on
the individual relations are independent of the join conditions~\cite{spark_cbo,postgres_cbo,selinger1979access,swami1994estimation}.

\subsection{Why not Query-driven Histograms}
\label{sec:prelim:drawback}

In this section, we briefly describe how query-driven histograms work~\cite{srivastava2006isomer,bruno2001stholes,markl2005consistently,kaushik2009consistent,aboulnaga1999self,lim2003sash,re2012understanding}, and then discuss their limitations, which motivate our work.

\ph{How Query-driven Histograms Work}
To approximate $f(x)$ (defined in \cref{prob:sec}), 
query-driven histograms adjust their bucket boundaries and bucket frequencies according to 
the queries they observe.
Specifically, they first determine bucket boundaries (\emph{bucket creation} step), 
and then compute their frequencies (\emph{training} step), as described next.

\begin{enumerate}
\item \emph{Bucket Creation}: Query-driven histograms determine their bucket boundaries based
on the given predicate's ranges~\cite{aboulnaga1999self,srivastava2006isomer,markl2005consistently,kaushik2009consistent,re2012understanding}. If the range of a later predicate overlaps with that of an earlier predicate, they split the bucket(s) created for the earlier one   into two or more buckets in order to ensure that the buckets do not overlap with one another.
\cref{fig:isomer_partition} shows an example of this bucket splitting operation.

\item \emph{Training}: 
After creating the buckets,   query-driven histograms assign frequencies to those buckets.
Early work~\cite{aboulnaga1999self,lim2003sash} 
determines bucket frequencies in the process of bucket creations. That is, when a bucket is split into two or more, the frequency of the original bucket is also split (or adjusted),
    such that it minimizes the estimation error for the latest observed query.
 
However, since this process does not minimize the (average) error across multiple queries, 
    their estimates are much less accurate.
 More recent work~\cite{srivastava2006isomer,markl2005consistently,kaushik2009consistent,re2012understanding} has addressed this limitation by
explicitly solving an optimization problem based on the maximum entropy principle.
That is, they search for   bucket frequencies that maximize the \emph{entropy} of the distribution while remaining consistent with the actual selectivities observed.
\end{enumerate}

\noindent 
 Although using the maximum entropy principle will lead to highly accurate estimates, 
    it still suffers from two key limitations.

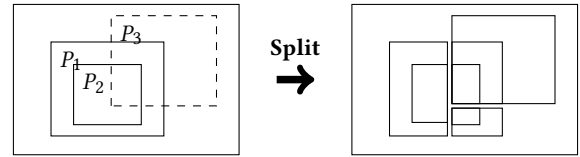
\begin{figure}[tb]
    \centering
    \begin{tikzpicture}
    
    \draw[] (0,0) rectangle (3,2);
    \draw[] (0.5,0.25) rectangle (2,1.5);   
    \draw[] (0.8,0.4) rectangle (1.7,1.2);  
    \draw[dashed] (1.3,0.65) rectangle (2.7,1.85);  
    \node[anchor=north west,font=\small] at (0.5,1.5) {$P_1$};
    \node[anchor=north west,font=\small] at (0.8,1.2) {$P_2$};
    \node[anchor=north west,font=\small] at (1.3,1.85) {$P_3$};
    
    \draw[->,line width=1.0mm] (3.5,1) -- (4.0,1) 
      node[midway,yshift=4mm,font=\small\bf] {Split};
    
    \def\r{4.5};
    \def\g{0.03};
    \draw[] (\r,0) rectangle (3+\r,2);      
    
    \draw[] (\r+0.5,0.25) rectangle (\r+1.3-\g, 1.5);
    \draw[] (\r+0.8,0.4+\g) rectangle (\r+1.3-\g, 1.2);
    
    \draw[] (\r+1.3+\g,0.25) rectangle (\r+2,0.65-\g);
    \draw[] (\r+1.3+\g,0.4) rectangle (\r+1.7, 0.65-\g);
    
    \draw[] (\r+1.3+\g,0.65+\g) rectangle (\r+2.7,1.85);  
    \draw[] (\r+1.3+\g,0.65+\g) rectangle (\r+1.7,1.2);
    \draw[] (\r+1.3+\g,0.65+\g) rectangle (\r+2,1.5);
    
    \end{tikzpicture}
    
\caption{Bucket creation for query-driven histograms. $P_3$ is the range of a new predicate. The existing buckets (for $P_1$ and $P_2$) are split into multiple buckets. The total number of buckets may grow \emph{exponentially} as more queries are observed.}
\label{fig:isomer_partition}

\vspace{-4mm}

\end{figure}

\ph{Limitation 1: Exponential Number of Buckets}
Since existing buckets may split into multiple ones for each new observed query, 
the number of buckets can potentially grow exponentially as the number of observed queries grows.
For example, in our experiment in \cref{sec:exp:model},
    the number of buckets was 22,370 for 100 observed queries, and 318,936 for 300 observed queries. 
Unfortunately, the number of buckets directly affects the training time. 
Specifically, using \emph{iterative scaling}---the optimization algorithm used by all  previous work~\cite{srivastava2006isomer,markl2005consistently,kaushik2009consistent,markl2007consistent,re2012understanding}---
    the cost of each iteration grows linearly with the number of variables (i.e., the number of buckets).
    This means that the cost of each iteration can grow exponentially with the number of observed queries.

As stated in \cref{sec:intro}, we address this problem by   employing a \emph{mixture model}, which can express  a probability distribution more effectively than query-driven histograms.
Specifically, our empirical study in \cref{sec:exp:model} shows that---using the same number of parameters---a mixture model achieves considerably more accurate estimates than histograms.

\ph{Limitation 2: Non-trivial Bucket Merge/Pruning}
Gi-ven that query-driven histograms~\cite{srivastava2006isomer,markl2005consistently} quickly 
become infeasible due to their large number of buckets,
    one might consider merging or pruning the buckets in an effort to
    reduce their training times.
However, merging or pruning the histogram buckets violates the assumption used by their optimization algorithms, i.e., iterative scaling.
Specifically, iterative scaling relies on the fact that a bucket is either \emph{completely included} in a query's predicate range or \emph{completely outside} of it.\footnote{For example, this property is required for the transition from Equation (6) to Equation (7) in \cite{markl2005consistently}.} That is, no partial overlap is allowed. This property must hold for each of the $n$ 
 predicates.
However, merging some of the buckets will inevitably cause  partial overlaps 
(between   
predicate  and histogram buckets). 
For interested readers, 
we have included a more detailed explanation of 
 why iterative scaling requires this assumption 
 in ~\cref{sec:appendix}.

\section{QuickSel: Model}
\label{sec:method}
\label{sec:model}

This section presents how \method models the population distribution and estimates the selectivity of a new query. \method's model relies on a probabilistic model called a \emph{mixture model}.
In \cref{sec:method:model}, we describe the mixture model employed by \method. \cref{sec:method:est} describes how to estimate the selectivity of a query using the mixture model. 
\cref{sec:model:sub} describes the details of \method's mixture model construction.

\input{fig_tab/fig_workload_kernels}

\subsection{Uniform Mixture Model}
\label{sec:method:model}

A mixture model is a probabilistic model that expresses a (complex) probability density function (of the population) as a combination of (simpler) probability density functions (of \emph{subpopulations}). The population distribution is the one that generates the tuple $x$ of $T$. The subpopulations are internally managed by \method to best approximate $f(x)$.

\ph{Uniform Mixture Model}
\method uses a type of mixture model, called \emph{the uniform mixture model}.
The uniform mixture model represents a population distribution $f(x)$ as a weighted summation of multiple uniform distributions, $g_z(x)$ for $z = 1, \ldots, m$. Specifically,
\begin{align}
f(x) = \sum_{z=1}^m \, h(z) \, g_z(x) = \sum_{z=1}^m \, w_z \, g_z(x)
\label{eq:mm}
\end{align}
where $h(z)$ is a categorical distribution that determines the weight of the $z$-th subpopulation, and $g_z(x)$ is the probability density function (which is a uniform distribution) for the $z$-th subpopulation. The support of $h(z)$ is the integers ranging from 1 to $m$; $h(z) = w_z$. The support for $g_z(x)$ is represented by a hyperrectangle $G_z$. Since $g_z(x)$ is a uniform distribution, $g_z(x) = 1 / |G_z|$ if $x \in G_z$ and 0 otherwise. The locations of $G_z$ and the values of $w_z$ are determined in the training stage (\cref{sec:training}). In the remainder of this section (\cref{sec:model}), we assume that $G_z$ and $w_z$ are given.

\ph{Benefit of Uniform Mixture Model}
The uniform mixture model was studied early in the statistics community~\cite{craigmile1997parameter,gupta1978uniform}; however, recently, a more complex model called \emph{the Gaussian mixture model} has received more attention~\cite{ml-textbook,yang2015compressive,ragothaman2016unsupervised}.\footnote{There are other variants of mixture models~\cite{moser2015simultaneous,asparouhov2016structural}.} The Gaussian mixture model uses a Gaussian distribution for each subpopulation; the smoothness of its probability density function (thus, differentiable) makes the model more appealing when gradients need to be computed. Nevertheless, we \emph{intentionally} use the uniform mixture model for \method due to its computational benefit in the training process, as we describe below.

As will be presented in \cref{sec:method:solving}, \method's training involves the computations of the intersection size between two subpopulations, for which the essential operation is evaluating the following integral:
\[
\int \, g_{z1}(x) \, g_{z2}(x) \; dx
\]
Evaluating the above expression for multivariate Gaussian distributions, e.g., $g_{z1}(x) = \exp \left( \minus x^\top \Sigma^{\minus 1} x \right) / \sqrt{(2 \pi)^d\, |\Sigma|}$, requires numerical approximations~\cite{genz1992numerical,joe1995approximations}, which are either slow or inaccurate. In contrast, the intersection size between two hyperrectangles can be exactly computed by simple min, max, and multiplication operations.

\subsection{Selectivity Estimation with UMM}
\label{sec:method:est}

For the uniform mixture model, computing the selectivity of a predicate $P_i$ is straightforward:

\begin{align*}
\int_{B_i} & \, f(x) \, dx
= \int_{B_i}\, \sum_{z=1}^m\, w_z \, g_z(x) \, dx
= \sum_{z=1}^m \, w_z \, \int_{B_i}\, g_z(x) \, dx \\
&= \sum_{z=1}^m \, w_z \, \int \, \frac{1}{|G_z|} \, I(x \in G_z \cap B_i) \, dx
= \sum_{z=1}^m \, w_z \, \frac{|G_z \cap B_i|}{|G_z|}
\end{align*}

Recall that both $G_z$ and $B_i$ are represented by hyperrectangles. Thus, their intersection is also a hyperrectangle, and computing its size is straightforward.

\ignore{
The integration, $\int_{P_i} g_j(x) dx$, can only be computed after determining $g_j(x)$.
For the reason we describe will below, \method uses a hyperrectangle $R_j$ for $g_j(x)$. Specifically, let $R_j = (l_1, u_1) \times \cdots \times (l_d, u_d)$ where $(l_i, u_i)$ is the range for the $i$-th dimension. We say $x \in R_j$ if $x \in [l_i, u_i]$ for all $i = 1, \ldots, d$. Then, $g_j(x) = 1$ if $x \in R_j$ and 0 otherwise.
Let $P_i \cap R_j$ be the intersection between $P_i$ and $R_j$, and let $|P_i \cap R_j|$ be the size of the intersection. Since $\int_{P_i} g_j(x) dx = |P_i \cap R_j| / |R_j|$, the selectivity for $P_i$ is simplified to:
\begin{align*}
\int_{P_i} \, f(x) \, dx = \sum_{j=1}^m \, w_j \, \frac{|P_i \cap R_j|}{|R_j|}
\end{align*}
}

\subsection{Subpopulations from Observed Queries}
\label{sec:model:sub}

We describe \method's approach to determining the boundaries of $G_z$ for $z = 1, \ldots, m$. 
Note that determining $G_z$ is orthogonal to the model training process, which we describe in \cref{sec:training}; thus, even if one devises an alternative approach to creating $G_z$, our fast training method is still applicable.

\method creates $m$ hyperrectangular ranges\footnote{The number $m$ of subpopulations is set to $\min(4 \cdot n, 4,000)$, by default.} (for the supports of its subpopulations) in a way that satisfies the following simple criterion: if more predicates involve a point $x$, use a larger number of subpopulations for $x$.
Unlike query-driven histograms, \method can easily pursue this goal by exploiting the property of a mixture model: the supports of subpopulations may overlap with one another.

In short, \method generates multiple points (using predicates) that represent the query workloads and create hyperrectangles that can sufficiently cover those points.
Specifically, we propose two approaches for this:  a sampling-based one and a clustering-based one. The sampling-based approach is faster; the clustering-based approach is more accurate. Each of these is described in more detail below.

\ph{Sampling-based}
This approach performs the following operations for creating $G_z$ for $z = 1, \ldots, m$.
\begin{enumerate}
\item 
Within each predicate range, generate multiple random points $r$.
Generating a large number of random points increases the consistency; however, \method limits the number to 10 since having more than 10 points did not improve accuracy in our preliminary study.
\item Use simple random sampling to reduce the number of points to $m$, which serves as the centers of $G_z$ for $z = 1, \ldots, m$.
\item 
The length of the $i$-th dimension of $G_z$ is set to twice the average of the distances (in the same $i$-th dimension) to the 10 nearest-neighbor centers.
\end{enumerate}

\noindent 
\cref{fig:model:creation} illustrates how the subpopulations are created using both (1) highly-overlapping query workloads and (2) scattered query workloads. In both cases, \method generates random points to represent the distribution of query workloads, which is then used to create $G_z$ ($z = 1, \ldots, m$), i.e., the supports of subpopulations. 
This sampling-based approach is faster, but it does not ensure the coverage of all random points $r$. In contrast, the following clustering-based approach ensures that.

\ph{Clustering-based}
The second approach relies on a clustering algorithm for generating hyperrectangles:
\begin{enumerate}
\item Do the same as the sampling-based approach.

\item Cluster $r$ into $m$ groups. (We used K-means++.)

\item For each of $m$ groups, we create the smallest hyperrectangle $G_z$ that covers all the points belonging to the group.
\end{enumerate}

\noindent Note that since each $r$ belongs to a cluster and
we have created a hyperrectangle that fully covers each cluster,
the union of the hyperrectangles covers all $r$.
%
Our experiments primarily use the sampling-based approach due to its efficiency, but we also compare them empirically in \cref{sec:exp:opt}.

\vspace{2mm}
The following section describes how to assign the weights (i.e., $h(z) = w_z$) of these subpopulations.

\section{QuickSel: Model Training}
\label{sec:method:optimal}
\label{sec:training}


This section describes how to compute the weights $w_z$ of \method's subpopulations.
For training its model, \method finds the model that maximizes uniformity while being consistent with the observed queries. In \cref{sec:training:prob}, we formulate an optimization problem based on this criteria. 
Next, \cref{sec:method:solving} presents how to solve the problem efficiently.

\subsection{Training as Optimization}
\label{sec:training:prob}

This section formulates an optimization problem for \method's training.
Let $g_0(x)$ be the uniform distribution with support $B_0$; that is, $g_0(x) = 1 / |B_0|$ if $x \in B_0$ and 0 otherwise. \method aims to find the model $f(x)$, such that the difference between $f(x)$ and $g_0(x)$ is minimized while being consistent with the observed queries.

There are many metrics that can measure the distance between two probability density functions $f(x)$ and $g_0(x)$, such as the earth mover's distance~\cite{rubner2000earth}, Kullback-Leibler divergence~\cite{kullback1951information}, the mean squared error (MSE), the Hellinger distance, and more. Among them, \method uses MSE (which is equivalent to $L2$ distance between two distributions) since it enables the reduction of our originally formulated optimization problem (presented shortly; \cref{prob:opt}) to a quadratic programming problem, which can be solved efficiently by many off-the-shelf optimization libraries~\cite{andersen2013cvxopt,apache_commons_opt,joptimizer,matlab_quadprog}. Also, minimizing MSE between $f(x)$ and $g_0(x)$ is closely related to maximizing the entropy of $f(x)$~\cite{srivastava2006isomer,markl2005consistently,kaushik2009consistent,re2012understanding}. 
See \cref{sec:rel_to_max_entropy} for the explanation of this relationship.

MSE between $f(x)$ and $g_0(x)$ is defined as follows:
\[
\text{MSE}(f(x), g_0(x)) = \int \, \left( f(x) - g_0(x) \right)^2 \, dx
\]
Recall that the support for $g_0(x)$ is $B_0$.
Thus, \method obtains the optimal weights by solving the following problem.

\vspace{2mm}

\begin{problem}[\method's Training]
\label{prob:opt}
\method obtains the optimal parameter $\ww$ for its model by solving:
\begin{align}
\argmin_{\ww}\; & \int_{x \in B_0} \, 
  \left( f(x) - \frac{1}{|B_0|} \right)^2 \; dx \label{eq:prob:density} \\
\text{such that } & \int_{B_i} f(x) \; dx = s_i \; \text{ for } i = 1, \ldots, n \\
& f(x) \ge 0 \label{eq:prob:cond}
\end{align}
Here, (\ref{eq:prob:cond}) ensures  $f(x)$ is a proper probability density function.
\end{problem}


\ignore{
In principle, there is no particular reason to prefer one to another; however, we use a variant of Kullback-Leibler divergence because it enables the conversion to a \emph{quadratic programming}~\cite{??} (\cref{thm:quadratic}), and thus efficient model training is possible. Also, minimizing our distance metric is closely related to maximizing \emph{entropy} (\cref{sec:disc:entropy}), which is the approach used by several previous works~\cite{??}.

Specifically, the variant of Kullback-Leibler (KL) divergence we use is expressed as follows:
\begin{align*}
\underbrace{\int\, f(x) \, \log \left( \frac{f(x)}{u(x)} \right) \; dx}_\text{KL divergence}
&\approx
\underbrace{\int\, f(x) \, \left( \frac{f(x)}{u(x)} - 1 \right) \; dx}_\text{Approximate KL divergence}
\end{align*}
The left-hand side of the above equation is the standard definition of the KL divergence and the right-hand side is our distance metric. The approximation is made using the first-order Taylor expansion: $\log x = x - 1$.
}

\subsection{Efficient Optimization}
\label{sec:method:solving}

We first describe the challenges in solving \cref{prob:opt}. Then, we describe how to overcome the challenges.

\ph{Challenge}
Solving \cref{prob:opt} in a na\"ive way is computationally intractable. For example, consider a mixture model consisting of (only) two subpopulations represented by $G_1$ and $G_2$, respectively. Then, $\int_{x \in B_0} (f(x) - g_0(x))^2 \,dx$ is:
\begin{align*}
&\int_{x \in G_1 \cap G_2} 
  \left( \frac{w_1 + w_2}{|G_1 \cap G_2|} - g_0(x) \right)^2 \, dx \\
+ & \int_{x \in G_1 \cap \neg G_2} 
  \left( \frac{w_1}{|G_1 \cap \neg G_2|} - g_0(x) \right)^2 \, dx \\
+ & \int_{x \in \neg G_1 \cap G_2} 
  \left( \frac{w_2}{|\neg G_1 \cap G_2|} - g_0(x) \right)^2 \, dx \\
+ & \int_{x \in \neg G_1 \cap \neg G_2} 
  \left( \frac{0}{|\neg G_1 \cap \neg G_2|} - g_0(x) \right)^2 \, dx \\
\end{align*}
Observe that with this approach, we need four separate integrations only for two subpopulations.
In general, the number of integrations is $O(2^m)$, which is $O(2^n)$.
Thus, this direct approach is computationally intractable.

\ph{Conversion One: Quadratic Programming}
\cref{prob:opt} can be solved efficiently by exploiting the property of the distance metric of our choice (i.e., MSE) and the fact that we use uniform distributions for subpopulations (i.e., UMM). 
The following theorem presents the efficient approach.


\vspace{2mm}

\begin{theorem}
\label{thm:quadratic}
The optimization problem in \cref{prob:opt} can be solved by the following quadratic optimization:
\begin{align*}
\argmin_{\ww} \quad & \ww^\top Q\, \ww \\
\text{such that} \quad & A \ww = \ss, \quad
  \ww \succcurlyeq \bm{0}
\end{align*}
where
\begin{align*}
(Q)_{ij} &= \frac{|G_i \cap G_j|}{|G_i| |G_j|}
\qquad
(A)_{ij} = \frac{|B_i \cap G_j|}{|G_j|}
\end{align*}
The bendy inequality sign ($\succcurlyeq$) means that every element of the vector on the left-hand side is equal to or larger than the corresponding element of the vector on the right-hand side.
\end{theorem}

\begin{proof}
This theorem can be shown by substituting the definition of \method's model (\cref{eq:mm}) into the probability density function $f(x)$ in \cref{eq:prob:density}.
Note that minimizing $(f(x) - g_0(x))^2$ is equivalent to minimizing $f(x) \, (f(x) - 2 \, g_0(x))$, which is also equivalent to minimizing $(f(x))^2$ since $g_0(x)$ is constant over $B_0$ and $\int f(x) \, dx = 1$.

The integration of $(f(x))^2$ over $B_0$ can be converted to a matrix multiplication, as shown below:
\begin{align*}
\int (f(x))^2  \; dx 
  &= \int \left[ \sum_{z=1}^m \frac{w_z \, I(x \in G_z)}{|G_z|}  \right]^2 
  \; dx \\
  &= \int \sum_{i=1}^m \sum_{j=1}^m \frac{w_i w_j}{|G_i| |G_j|} \,
      I(x \in G_i) \, I(x \in G_j) \; dx 
\end{align*}
which can be simplified to
\begin{align*}
  &\sum_{i=1}^m \sum_{j=1}^m \frac{w_i w_j}{|G_i| |G_j|} \, |G_i \cap G_j|  \\
  &= \begin{bmatrix}
       w_1 \\
       w_2 \\
       \vdots \\
       w_m
     \end{bmatrix}^\top
     \begin{bmatrix}
       \frac{|G_1 \cap G_1|}{|G_1| |G_1|}  & \cdots & \frac{|G_1 \cap G_m|}{|G_1| |G_m|} \\
       \vdots & & \vdots \\
       \frac{|G_m \cap G_1|}{|G_m| |G_1|}  & \cdots & \frac{|G_m \cap G_m|}{|G_m| |G_m|}
     \end{bmatrix}
     \begin{bmatrix}
       w_1 \\
       w_2 \\
       \vdots \\
       w_m
     \end{bmatrix} \\
  &= \ww^\top Q \, \ww
\end{align*}

Second, we express the equality constraints in an alternative form. Note that the left-hand side of each equality constraint, i.e., $\int_{B_i} f(x) \, dx$, can be expressed as:
\begin{align*}
&\int_{B_i} f(x) \; dx 
= \int_{B_i} \, \sum_{j=1}^m \, \frac{w_j}{|G_j|} \, I(x \in G_j) \; dx \\
&= \sum_{j=1}^m \, \frac{w_j}{|G_j|} \, \int_{B_i} I(x \in G_j) \; dx 
= \sum_{j=1}^m \, \frac{w_j}{|G_j|} \, |B_i \cap G_j| \\
&= \begin{bmatrix}
       \frac{|B_i \cap G_1|}{|G_1|} &
       \cdots &
       \frac{|B_i \cap G_m|}{|G_m|}
   \end{bmatrix}
   \begin{bmatrix}
       w_1 \\
       \vdots \\
       w_m
   \end{bmatrix}
\end{align*}
Then, the equality constraints, i.e., $\int_{B_i} f(x) \, dx = s_i$ for $i = 1, \ldots, n$, can be expressed as follows:
\begin{align*}
&\begin{bmatrix}
       \frac{|B_1 \cap G_1|}{|G_1|} &
       \cdots &
       \frac{|B_1 \cap G_m|}{|G_m|} \\
       \vdots & \ddots & \vdots \\
       \frac{|B_n \cap G_1|}{|G_1|} &
       \cdots &
       \frac{|B_n \cap G_m|}{|G_m|}
   \end{bmatrix}
   \begin{bmatrix}
       w_1 \\
       \vdots \\
       w_m
   \end{bmatrix}
   = 
   \begin{bmatrix}
       s_1 \\
       \vdots \\
       s_m
   \end{bmatrix} \\
&\Rightarrow \quad A \ww = \ss
\end{align*}

Finally, $\ww^\top \bm{1} = 1$ if and only if $\int f(x) = 1$, and $\ww \succcurlyeq \bm{0}$ for arbitrary $G_z$ if and only if $\int f(x) \ge 0$.
\end{proof}

The implication of the above theorem is significant: we could reduce the problem of $O(2^n)$ complexity to the problem of only $O(n^2)$ complexity.

\ph{Conversion Two: Moving Constraints}
The quadratic programming problem in \cref{thm:quadratic} can be solved efficiently by most off-the-shelf optimization libraries; however, we can solve the problem even faster by converting the problem to an alternative form. We first present the alternative problem, then discuss it.

\vspace{2mm}

\begin{problem}[\method's QP]
\label{prob:opt:alter}
\method solves this problem alternative to the quadratic programming problem in
\cref{thm:quadratic}:
\begin{align*}
\argmin_{\ww} \quad \ell(\ww) = \ww^\top Q\, \ww + \lambda\, \| A \ww - \ss \|^2
\end{align*}
where
$\lambda$ is a large real value (\method uses $\lambda = 10^6$).
\end{problem}

In formulating \cref{prob:opt:alter}, two types of conversions are performed: (1) the consistency with the observed queries 
(i.e., $A \ww = \ss$) is moved into the optimization objective as a penalty term, and (2) the positivity of $f(x)$ is not explicitly specified (by $\ww \succcurlyeq \bm{0}$). These two types of conversions have little impact on the solution for two reasons. First, to guarantee the consistency, a large penalty (i.e., $\lambda = 10^6$) is used. Second, the mixture model $f(x)$ is bound to approximate the true distribution, which is always non-negative. We empirically examine the advantage of solving \cref{prob:opt:alter} (instead of solving the problem in \cref{thm:quadratic} directly) in \cref{sec:exp:approx}.

The solution $\ww^*$ to \cref{prob:opt:alter} can be obtained in a straightforward way by setting its gradients of the objective (with respect to $\ww$) equal to $\bm{0}$:
\begin{align*}
& \frac{\partial \ell(\ww^*)}{\partial \ww} 
= 2 \, Q \, \ww^* + 2 \, \lambda \, A^\top (A \, \ww^* - \ss) = \bm{0} \\
&\Rightarrow \ww^* = \left( Q + \lambda\, A^\top A \right)^{\minus 1} \, \lambda \, A \, \ss
\end{align*}
Observe that $\ww^*$ is expressed in a closed form; thus, we can obtain $\ww^*$ analytically instead of using iterative procedures typically required for general quadratic programming.





\section{Experiment}
\label{sec:exp}

In this section, we empirically study \method. 
\ignore{First, we compare \method's end-to-end selectivity estimation quality against query-driven histograms (\cref{sec:exp:acc}). 
Second, we study the efficiency of \method's mixture model (\cref{sec:exp:model}).
Third, we compare \method's optimization efficiency to a standard solution (\cref{sec:exp:opt}). 
Fourth, we compare \method to the methods that rely on periodic database scans (\cref{sec:exp:others}).}
In summary, our results 
show the  following:

\begin{enumerate}
\item \textbf{End-to-end comparison against other query-driven methods:}
\method was 
significantly faster (34.0$\times$--179.4$\times$) for the same accuracy---and produced much more accurate estimates (26.8\%--91.8\% lower error) for the same time limit---than previous query-driven methods. (\cref{sec:exp:acc})

\item \textbf{Comparison against periodic database scans:} For the same storage size, \method's selectivity estimates were  77.7\% and 91.3\% more accurate than   scan-based histograms and sampling, respectively.  (\cref{sec:exp:others})

\item \textbf{Impact on PostgreSQL performance:}
Using \method for PostgreSQL makes the system
2.25$\times$ faster (median) than the default.
(\cref{sec:exp:impact})

\item \textbf{Effectiveness of \method's mixture model:}
\method's model produced considerably more accurate estimates than histograms given the same number of parameters. (\cref{sec:exp:model})

\item \textbf{Robustness to workload shifts:}
\method's accuracy quickly recovers after sudden workload shifts.
(\cref{sec:exp:workload})

\item \textbf{Optimization efficiency:} \method's optimization method (\cref{prob:opt:alter}) 
was 1.5$\times$--17.2$\times$ faster than solving the standard quadratic programming. (\cref{sec:exp:opt})


\end{enumerate}

\subsection{Experimental Setup}

\ph{Methods}
Our experiments compare \method to six other selectivity estimation methods.


\vspace{1mm}
\noindent \textbf{Query-driven Methods:}
\begin{enumerate}
\item STHoles~\cite{bruno2001stholes}: This method creates histogram buckets by partitioning existing buckets (as in \cref{fig:isomer_partition}). The frequency of an existing bucket is distributed uniformly among the newly created buckets.

\item ISOMER~\cite{srivastava2006isomer}: 
This method applies STHoles for histogram bucket creations, but it computes the optimal frequencies of the buckets by finding the maximum entropy distribution.
Among existing query-driven methods, ISOMER produced the highest accuracy in our experiments.

\item ISOMER+QP: This method combines ISOMER's approach for creating histogram buckets and \method's quadratic programming (\cref{prob:opt:alter}) for computing the optimal bucket frequencies. 

\item QueryModel~\cite{anagnostopoulos2015learning}: 
This method computes the selectivity estimate by a weighted average of the selectivities of observed queries. The weights are determined based on the similarity of the new query and each of the queries observed in the past.
\end{enumerate}

\vspace{1mm}
\noindent \textbf{Scan-based Methods:}
\begin{enumerate}
\setcounter{enumi}{4}
\item AutoHist: This method creates an equiwidth multidimensional histogram by scanning the data. It also updates its histogram whenever more than 20\% of the data changes (this is the default setting with
SQL Server's \texttt{AUTO\_UPDATE\_} \texttt{STATISTICS} option~\cite{sqlserver_auto}).

\item AutoSample: This method relies on a uniform random sample of data to estimate selectivities.
Similar to AutoHist, AutoSample updates its sample whenever more than 10\% of the data changes. \barzan{any citations? any real dbms uses this?} \yongjoo{Not found yet}
\end{enumerate}

\noindent We have implemented all methods in Java.

\input{fig_tab/fig_exp_performance}

\ph{Datasets and Query Sets}
We use two real datasets and one synthetic dataset in our
 experiments, as follows:
\begin{enumerate}
\item \dmv: This dataset contains the vehicle registration records of New York State~\cite{dmv_dataset}. It contains 11,944,194 rows.
Here, the queries ask for the number of  valid registrations for vehicles produced within a certain date range.
Answering these queries involves predicates on
 three attributes: \texttt{model\_year}, \texttt{registration\_date}, and \texttt{expiration\_date}.
 
\item \insta: This dataset contains the sales records of an online grocery store~\cite{instacart_dataset}. 
We use their \texttt{orders} table, which contains 3.4 million sales records.
Here, the queries ask for the reorder frequency for
orders made during different hours of the day. Answering these queries involves
predicates on
 two attributes: \texttt{order\_hour\_of\_day} and \texttt{days\_since\_prior}.
(In \cref{sec:exp:others}, we use more attributes (up to ten).)
 
\item \synthetic: We also generated a synthetic dataset using a bivariate dimensional normal distribution. 
We varied this dataset to study our method under workload shifts,
different degrees of correlation between the attributes, and more. 
Here, the queries count the number of points that lie within a randomly generated rectangle.
\end{enumerate}

\noindent For each dataset, we measured the estimation quality   using   100 test queries
not used for training.
The ranges for selection predicates (in queries) were generated randomly within a feasible region; the ranges of different queries may or may not overlap.

\ph{Environment}
All our experiments were performed on   \texttt{m5.4xlarge}   EC2  instances, with 16-core Intel Xeon 2.5GHz and 64 GB of memory running Ubuntu 16.04.

\ph{Metrics}
We use
the root mean square (RMS) error:
\[
\text{RMS error} = 
    \left(
        \frac{1}{t} \sum_{i=1}^t
        (\text{true\_sel} - \text{est\_sel})^2
    \right)^{1/2}
\]
where $t$ is the number of test queries. We report the RMS errors in percentage (by treating both true\_sel and est\_sel as percentages).


When reporting training time, we include the time required for refining a model using an additional observed query,
which itself includes the time to store the query and run the necessary optimization routines.

\input{fig_tab/fig_exp_other}

\subsection{Selectivity Estimation Quality}
\label{sec:exp:acc}

In this section, we compare the end-to-end selectivity estimation quality of 
\method versus query-driven histograms. 
%
\ignore{
We have summarized the main results in \cref{tab:exp:acc}.
The table reports that, for both \dmv and \insta dataset, \method was significantly faster for the same accuracy, and significantly more accurate for the same time limit.
}
Specifically, we gradually increased the number of observed queries provided to each method from 10 to 1,000. 
For each number of observed queries, 
we measured the estimation error and training time of each method using 100 test queries.

 These results are reported 
in \cref{fig:exp:acc}. 
Given the same number of observed queries,  
\method's training was significantly faster
 (\cref{fig:exp:perf:a,fig:exp:perf:d}),
 while still achieving comparable 
estimation errors 
(\cref{fig:exp:perf:b,fig:exp:perf:e}).
We also studied the relationship between errors and training times in \cref{fig:exp:perf:c,fig:exp:perf:f}, confirming 
\method's superior efficiency (STHoles, ISOMER+QP, and QueryModel are
omitted 
in these figures due to their poor performance). 
In summary, \method was able to quickly learn
from a large
 number of observed queries (i.e., shorter training time)
    and produce highly accurate models.

\subsection{Comparison to Scan-based Methods}
\label{sec:exp:others}


We also compared \method to two automatically-updating scan-based methods, AutoHist and AutoSample,
which incorporate SQL Server's automatic updating rule into equiwidth multidimensional histograms and samples, respectively. Since both methods incur an up-front cost for obtaining
their statistics, they should produce relatively more accurate estimates initially (before seeing new queries).
In contrast, the accuracy of \method's estimates 
    should quickly improve as new queries are observed.

To verify this empirically, 
    we first generated a \synthetic dataset (1 million tuples) with correlation 0. We then inserted 200K new tuples generated from a distribution with a \emph{different} correlation 
    after processing   200 queries, and repeated this process. 
    In other words, after processing the first 100 queries, we inserted new data with correlation 0.1; after processing the next 100 queries, we inserted new data with correlation 0.2; and continued this process until a total of 1000 queries were processed. 
    We performed this process for each method under comparison. 
    \method adjusted
     its model each time after observing 100 queries.
    AutoHist and AutoSample updated their statistics after each batch of
    data insertion. 
    \method and AutoHist both used 100 parameters (\# of subpopulations for the mixture model and \# of buckets for histograms); 
    AutoSample used a sample of 100 tuples.

\input{fig_tab/fig_exp_performance_impact2}

\input{fig_tab/fig_exp_model}

\cref{fig:exp:nonsl} shows the error of each method. 
As expected, AutoHist produced more accurate estimates initially.
However, as more queries were processed, 
    the error of \method drastically decreased. 
 In contrast, the errors of AutoSample and AutoHist 
    did not improve with more queries, as they only depend on the frequency 
    at which a new scan (or sampling) is performed. 
  After processing only 100 queries (i.e., initial update), 
\method produced more accurate estimates than both AutoHist and AutoSample.  On average (including the first 100 queries), 
\method was 71.4\% and 89.8\% more accurate than AutoHist and AutoSample, respectively. This is consistent with the previously reported observations that
    query-driven methods yield better accuracy than
        scan-based ones~\cite{bruno2001stholes}.
(The reason why query-driven proposals have not  been widely adopted
to date is due to their prohibitive cost; see \cref{sec:related:queries}).

In addition, \cref{fig:exp:nonsl:time} compares the update times of the three methods.
By avoiding scans, \method's query-driven updates were 525$\times$ and 243$\times$ faster than AutoHist and AutoSample, respectively.

Finally, we studied how the performance of those methods changed as we increased the data dimension (i.e., the number of attributes appearing in selection predicates). First, using the \insta dataset,
we designed each query to target a random subset of dimensions (N/2) as increasing the dimension N from 2 to 10. 
In all test cases (\cref{fig:exp:scan:d} left), \method’s accuracy was consistent, showing its ability to scale to high-dimensional data. Also in this experiment, \method performed significantly better than, or comparably to, histograms and sampling. 
We could also obtain a similar result using the \synthetic dataset (\cref{fig:exp:scan:d} right).
This consistent performance across different data dimensions is primarily due to how \method is designed; that is, its estimation only depends on how much queries overlap with one another.





\subsection{Impact on Query Performance}
\label{sec:exp:impact}

This section examines \method's impact on query performance. That is, we test if \method's more accurate
selectivity estimates can lead to improved query performance for actual database systems (i.e., shorter latency).

To measure the actual query latencies, we used PostgreSQL ver. 10 with a third-party extension, called \texttt{pg\_hint\_plan}~\cite{pghint}.
Using this extension, we enforced our own estimates (for PostgreSQL's query optimization) in place of the default ones.
We compared PostgreSQL Default (i.e., no hint) and \method---to measure the latencies of the following join query in processing the \insta dataset:

\begin{lstlisting}[
    basicstyle=\small\ttfamily,
    xleftmargin=10pt
]
select count(*)
from S inner join T on S.tid = T.tid
       inner join U on T.uid = U.uid
where (range_filter_on_T)
  and (range_filter_on_S);
\end{lstlisting}

\noindent
where the joins keys for the tables \texttt{S}, \texttt{T}, and \texttt{U} were in the PK-FK relationship, as described by the schema (of \insta).

\cref{fig:exp:perf_impact} shows the speedups \method could achieve in comparison to PostgreSQL Default. Note that \method does not improve any underlying I/O or computation speed; its speedups are purely from helping PostgreSQL's query optimizer choose a more optimal plan based on improved selectivity estimates. Even so, \method could bring 2.25$\times$ median speedup, with 3.47$\times$ max speedup. In the worst case, PostgreSQL with \method was almost identical to PostgreSQL Default (i.e., 0.98$\times$ speedup).

\subsection{\method's Model Effectiveness}
\label{sec:exp:model}

In this section, we  compare the effectiveness of \method's model 
to that of models used in the previous work.
Specifically, the effectiveness is assessed by
(1) how the model size---its number of parameters---grows as the number of observed queries grows, 
and (2)
how quickly its error decreases as its number of  parameters grows.  

\ignore{Here, we use \method's default setting whereby   its number of parameters increases linearly with the number of observed queries, deferring the analysis of its non-default setting to \cref{sec:exp:param}.}

\begin{figure}[t]

\pgfplotsset{workloadfig/.style={
    width=65mm,
    height=32mm,
    xmin=0,
    xmax=300,
    ymin=0,
    ymax=0.04,
    xlabel=Query Sequence Number,
    ylabel=RMS Error,
    xlabel near ticks,
    ylabel near ticks,
    ylabel style={align=center},
    xtick={0, 50, ..., 300},
    ytick={0, 0.01, 0.02, 0.03, 0.04},
    yticklabels={0\%, 1\%, 2\%, 3\%, 4\%},
    scaled y ticks = false,
    ylabel shift=-2pt,
    xlabel shift=-2pt,
    legend style={
        at={(1.05,1.0)},anchor=north west,column sep=2pt,
        draw=black,fill=none,line width=.5pt,
        /tikz/every even column/.append style={column sep=5pt},
        font=\footnotesize,
    },
    legend cell align={left},
    legend columns=1,
    every axis/.append style={font=\footnotesize},
    ymajorgrids,
    minor grid style=lightgray,
}}

\centering

\begin{tikzpicture}
\begin{axis}[workloadfig,
    ]

\addplot[mark=*,vintageblack,mark size=1.5,thick]
table[x=x,y=y] {
x y
0	0.014043476
10	0.022072786
20	0.015459343
30	0.020557596
40	0.016907739
50	0.017329738
60	0.010130827
70	0.021577869
80	0.024287008
90	0.016009445
100	0.00322752
110	0.005056955
120	0.003456435
130	0.004727297
140	0.003890263
150	0.0040467
160	0.002195448
170	0.004959721
180	0.005570967
190	0.003646233
200	0.010634815
210	0.009980641
220	0.010558702
230	0.010103518
240	0.010481613
250	0.010366147
260	0.010820137
270	0.009972486
280	0.009710103
290	0.010530145
};

\addplot[mark=*,vintagegreen,mark size=1.5,thick]
table[x=x,y=y] {
x y
0	0.023958306
10	0.016758871
20	0.008928302
30	0.015565997
40	0.017910733
50	0.011619284
60	0.008225758
70	0.005794346
80	0.012295496
90	0.013667535
100	0.016482134
110	0.010148335
120	0.014675194
130	0.012298712
140	0.019657059
150	0.033598044
160	0.016878892
170	0.019591617
180	0.006728816
190	0.004050885
200	0.016375688
210	0.02084985
220	0.008518387
230	0.01162606
240	0.007547007
250	0.012578665
260	0.005321596
270	0.009461271
280	0.020385251
290	0.006630684
};

\addplot[mark=*,vintageorange,mark size=1.5,thick]
table[x=x,y=y] {
x y
0	0.029573072
10	0.001420434
20	0.000079971
30	0.000081535
40	0.000079646
50	0.000045865
60	0.00006189
70	0.000064898
80	0.000069354
90	0.000068426
100	0.011631817
110	0.0002922
120	0.000284525
130	0.000126782
140	0.000105779
150	0.000045742
160	0.000056393
170	0.000054067
180	0.000051621
190	0.000070682
200	0.016063209
210	0.000191805
220	0.000168858
230	0.00011363
240	0.000102185
250	0.000050302
260	0.000038317
270	0.000044829
280	0.000026292
290	0.000059506
};

\addlegendentry{Histograms}
\addlegendentry{Sampling}
\addlegendentry{\method}

\end{axis}
\end{tikzpicture}

\vspace{-4mm}

\caption{Robustness to sudden workload shifts, which occurred 
at the sequence \#100 and at \#200.
\method's error increased temporarily right after each workload jump,
but it reduced soon.}
\label{fig:exp:workload}

\end{figure}
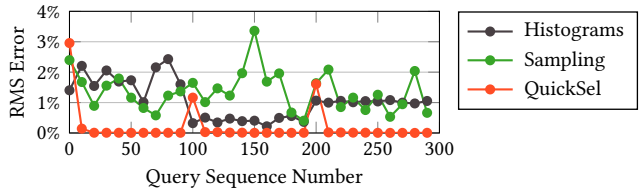

\begin{figure}[t]

\pgfplotsset{optfig/.style={
    width=62mm,
    height=32mm,
    xmin=0,
    xmax=1000,
    ymin=0,
    ymax=0.1,
    xlabel=Number of Observed Queries,
    ylabel=Runtime (ms),
    xlabel near ticks,
    ylabel near ticks,
    ylabel style={align=center},
    xtick={0, 200, ..., 1000},
    ytick={0, 0.02, 0.04, 0.06, 0.08, 0.1},
    yticklabels={0, 20, 40, 60, 80, 100},
    scaled y ticks = false,
    ylabel shift=-2pt,
    xlabel shift=-2pt,
    legend style={
        at={(1.05,1.0)},anchor=north west,column sep=2pt,
        draw=black,fill=none,line width=.5pt,
        /tikz/every even column/.append style={column sep=5pt},
        font=\footnotesize,
    },
    legend cell align={left},
    legend columns=1,
    every axis/.append style={font=\footnotesize},
    ymajorgrids,
    minor grid style=lightgray,
}}

\centering

\begin{tikzpicture}
\begin{axis}[optfig]

\addplot[mark=*,mark size=1.5,mark options={fill=vintageblack,draw=vintageblack},
thick,draw=vintageblack,
]
table[x=x,y=y] {
x y
20   0.0017
50   0.0017
100  0.0024
200  0.0050
300  0.0088
400  0.0154
500  0.0235
600  0.0326
700  0.0400
800  0.0531
900  0.0687
1000 0.0833
};

\addplot[mark=*,mark size=1.5,mark options={fill=vintageorange,draw=vintageorange},
thick,draw=vintageorange,
]
table[x=x,y=y] {
x y
20   0.0001
50   0.0001
100  0.0002
200  0.0005
300  0.0011
400  0.0018
500  0.0028
600  0.0038
700  0.0050
800  0.0065
900  0.0081
1000 0.0100
};

\addlegendentry{Standard QP}
\addlegendentry{\method's QP}

\end{axis}
\end{tikzpicture}

\vspace{-4mm}

\caption{\method's optimization effect.}
\label{fig:exp:opt}

\end{figure}
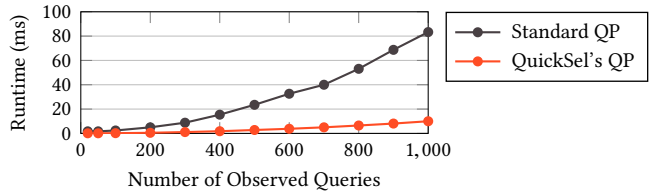

\cref{fig:exp:model:c} reports the relationship between the number of observed queries and the number of model parameters.   
As discussed in \cref{sec:prelim:drawback}, the number of buckets (hence, parameters) of ISOMER increased quickly as the number of observed queries grew. 
STHoles was able to keep the number of its parameters small due to its bucket merging technique; however, this had a negative impact on its accuracy.
 Here, \method used the least number of model parameters. For instance, when 100 queries were observed for \dmv, \method had 10$\times$ fewer parameters than STHoles and 56$\times$ fewer parameters than ISOMER.

 We also studied the relationship between  the 
 number of model parameters and the error.
 The lower the error (for the same number of model parameters), the more effective the model.
 \cref{fig:exp:model:d} shows the result. 
Given the same number of model parameters, \method produced significantly more accurate estimates. Equivalently, \method produced the same quality estimates with much fewer model parameters. 

\subsection{Robustness to Workload Shifts}
\label{sec:exp:workload}

In this section, we test QuickSel’s performance under significant workload shifts.
That is, after observing a certain number of queries (i.e., 100 queries) around a certain region of data, the query workload suddenly jumps to a novel region.  This pattern repeats several times.

\cref{fig:exp:workload} shows the result. Here, we could observe the following pattern. QuickSel’s error increased significantly right after each jump (i.e., at query sequence \#100 and at \#200), producing 1.5$\times$-3.6$\times$ higher RMS errors compared to histograms. However, QuickSel’s error dropped quickly, achieving 12$\times$-378$\times$ lower RMS errors than histograms. This was possible due to QuickSel’s faster adaptation.

\begin{figure}[t]

\pgfplotsset{workloadfig/.style={
    width=45mm,
    height=30mm,
    xmin=0.5,
    xmax=5.5,
    ymin=0,
    ymax=0.030,
    xlabel=Data Dimension,
    ylabel=RMS Error,
    xlabel near ticks,
    ylabel near ticks,
    ylabel style={align=center},
    xtick={1, 2, 3, 4, 5},
    xticklabels={2, 4, 6, 8, 10},
    ytick={0, 0.005, 0.01, 0.015, 0.020, 0.025, 0.030},
    yticklabels={0.0\%, 0.5\%, 1.0\%, 1.5\%, 2.0\%, 2.5\%, 3.0\%},
    scaled y ticks = false,
    ylabel shift=-2pt,
    xlabel shift=-2pt,
    legend style={
        at={(0.0,1.1)},anchor=south west,column sep=2pt,
        draw=black,fill=none,line width=.5pt,
        /tikz/every even column/.append style={column sep=5pt},
        font=\footnotesize,
    },
    legend cell align={left},
    legend columns=3,
    every axis/.append style={font=\footnotesize},
    ymajorgrids,
    minor grid style=lightgray,
    legend image code/.code={%
    \draw[#1, draw=none] (0cm,-0.1cm) rectangle (0.6cm,0.1cm);}
}}

\centering

\begin{subfigure}[b]{0.48\linewidth}
\begin{tikzpicture}
\begin{axis}[workloadfig,
    ybar,
    bar width=1mm,
    legend image code/.code={%
        \draw[#1, draw=none] (0cm,-0.1cm) rectangle (0.6cm,0.1cm);}
    ]

\addplot[fill=vintageblack,draw=none
]
table[x=x,y=y] {
x y
1 0.000027
2 0.000192
3 0.001234
4 0.006034
5 0.024774
};

\addplot[fill=vintageorange,draw=none]
table[x=x,y=y] {
x y
1 0.000075
2 0.000126
3 0.000768
4 0.001316
5 0.003954
};

\addlegendentry{Sampling-based}
\addlegendentry{Clustering-based}

\end{axis}
\end{tikzpicture}

\vspace{-2mm}
\caption{Accuracy}
\label{fig:exp:subpop:a}
\end{subfigure}
\hfill
\begin{subfigure}[b]{0.48\linewidth}
\begin{tikzpicture}
\begin{axis}[workloadfig,
    ybar,
    bar width=1mm,
    ylabel=Overhead (ms),
    ymin=0,
    ymax=0.020,
    ytick={0, 0.005, 0.010, 0.015, 0.020},
    yticklabels={0, 5, 10, 15, 20},
    legend image code/.code={%
        \draw[#1, draw=none] (0cm,-0.1cm) rectangle (0.6cm,0.1cm);}
    ]

\addplot[fill=vintageblack,draw=none
]
table[x=x,y=y] {
x y
1 0.008006516
2 0.008513297
3 0.008673624
4 0.009354242
5 0.010277341
};

\addplot[fill=vintageorange,draw=none
]
table[x=x,y=y] {
x y
1 0.007477261
2 0.008332779
3 0.011489065
4 0.012661078
5 0.014273459
};


\end{axis}
\end{tikzpicture}

\vspace{-2mm}
\caption{Per-query Overhead}
\label{fig:exp:subpop:b}
\end{subfigure}

\vspace{-2mm}
\caption{Subpopulation generation approaches.
Clustering-based was more accurate, but slower.}
\label{fig:exp:sub}
\end{figure}
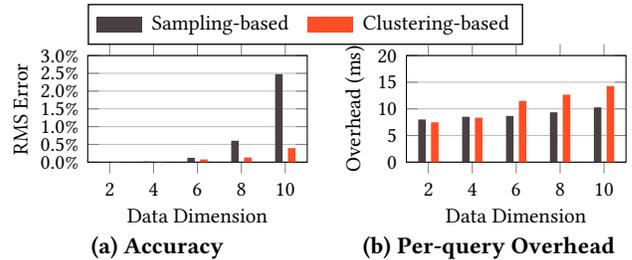

\subsection{\method Internal Mechanisms}
\label{sec:exp:approx}
\label{sec:exp:opt}

In this section, we empirically study (1) the effect of \method's optimization (presented in \cref{sec:method:solving}), 
and (2) two alternative mechanisms for generating subpopulations
(presented in \cref{sec:model:sub}).

\ph{Optimization Efficiency}
To study \method's optimization efficiency,
    we compared two approaches for solving the quadratic problem 
    defined in \cref{thm:quadratic}:
solving the original QP without 
    any modifications versus solving our modified 
version (\cref{prob:opt:alter}).
We used the  \texttt{cvxopt}
library for the former and used  \texttt{jblas} (a linear algebra library) for the latter. 
Both libraries use multiple cores for parallel processing. 

\cref{fig:exp:opt} shows
the time taken by   each optimization
approach. The second approach (\cref{prob:opt:alter})
was increasingly more 
efficient as the number of observed queries grew.
For example, it was 8.36$\times$ faster when the number of observed queries reached 1,000.
This is thanks to the modified problem 
having an analytical solution, while the original
problem required an iterative gradient descent 
solution.

\ph{Subpopulation Generation}
We empirically studied the two subpopulation generation approaches (i.e., the sampling-based approach and the clustering-based approach, \cref{sec:model:sub}) in terms of their scalability to high-dimensional data.
Specifically, we compared their estimation accuracies and computational overhead using the \synthetic dataset (with its dimension set to 2--10).

\cref{fig:exp:sub} reports the results. As shown in \cref{fig:exp:sub}(a), the clustering-based approach obtained impressive accuracy in comparison to the sampling-based one. However, as shown in \cref{fig:exp:sub}(b), the clustering-based approach produced higher overhead (i.e., longer training times),
which is an example of the natural tradeoff between cost and accuracy.

\ignore{
\subsection{QuickSel's Robustness}
\label{sec:exp:param}

\tofix{Remove this section}
 
We further studied how \method's accuracy is affected by data distribution, query workloads, number of  model parameters, and number of columns in the schema.

\ph{Data Correlation}
To study how \method's accuracy changes based on 
different degrees of correlations in the data,
    we used  our \synthetic workload,
    generating values with different correlations between columns.
    In each case, \method trained its model using 100 observed queries; the error was measured using the 
    other 100 queries (not used for training).
 As shown in \cref{fig:exp:param:a}, 
    the errors remained almost identical across
        all different degrees of correlation.
    

\ph{Workload Shifts}
We also studied  \method's accuracy 
under situations where the 
query workload shifts over time. 
Here, we generated the  \synthetic dataset
with
correlation 0.5.
We generated  queries with different predicates,
    with each predicate being a different rectangle in the 2-dimensional 
    space.\footnote{The purpose of this scenario
     is to 
        simulate a workload shift;
        we study high-dimensional data and 
        complex queries in \cref{fig:exp:param:d}.}
We simulated three scenarios of workload shifts
    by modifying these rectangular predicates. 
First, we created a \emph{random-shift} workload by choosing random rectangles in the space. We also created a \emph{sliding-shift} workload by gradually moving the rectangles from the left-tail of the normal
        distribution towards the right-tail.
Third, we created a \emph{no-shift} workload by using the same rectangle for all queries.

In each scenario, we first trained \method on the first 10 observed queries (i.e., sequence numbers: 1--10) and   measured the accuracy on the next 10 observed queries (i.e., sequence numbers: 11--20). 
Then, we trained \method on the first 20 observed queries (i.e., sequence numbers: 1--20) and measured the accuracy using the next 10 observed queries (i.e., their sequence numbers: 21--30). We kept increasing the number of observed queries until we reached 
    1,000 observed queries for training.
    
The results are plotted in \cref{fig:exp:param:b}. 
As expected, the errors were highest under random-shift. However, the error still decreased as \method observed more queries; after 100 queries, the relative error was only 1.2\% for the random shift-workload. The errors were lower in general for the other workloads.


\ph{Model Parameter Count}
To study the relationship between \method's 
number of parameters and  its accuracy, 
  we disabled \method's default setting
  (i.e., \# of model params = 4 $\times$ \# of observed queries).
Instead, we manually controlled its number of parameters.
 
\cref{fig:exp:param:c} shows the results for a \synthetic dataset with correlation 0.5. 
As expected, the errors were relatively higher when the number of model parameters was extremely small (i.e., 10). 
However, as soon as the number of model parameters reached 50, the errors
    were drastically reduced.

\begin{table*}[t]
\caption{Comparison of selectivity estimation methods}
\label{tab:related}

\vspace{-2mm}

\centering
\footnotesize
\renewcommand{\arraystretch}{1.2}
\begin{tabular}{p{18mm} p{24mm} p{30mm} p{90mm}}
\toprule
\textbf{Approach} & \textbf{Model}  & \textbf{Method} & \textbf{Key Contributions} \\
\midrule
\multirow{5}{*}{\parbox{16mm}{\textbf{Based on} \\ \textbf{Database} \\ \textbf{Scans} \\ (Scan-based \\ Selectivity \\ Estimation) }}
  & 
  Histograms
    & Multi-dim Hist~\cite{lynch1988selectivity,deshpande2001independence,ilyas2004cords} & Introduces multidimensional histograms \\
  & & Muralikrishna~\cite{muralikrishna1988equi} & Introduces equidepth histograms \\
  & & Van Gelder~\cite{van1993multiple} & Estimates Join selectivity with histograms for \emph{important} domains \\
  & & GOH~\cite{jagadish2001global} & Optimizes single-attribute histograms for joint distribution \\
  & & Thaper~\cite{thaper2002dynamic} & Builds histograms over streaming data \\
  & & To~\cite{to2013entropy} & Builds histograms with entropy as a metric \\
\cmidrule{2-4}
  & Sampling 
    & Lipton~\cite{lipton1990practical} & Introduces adaptive sampling for high accuracy \\
  & & Haas~\cite{haas1994relative} & Uses sampling for join selectivity estimation \\
  & & Riondato~\cite{riondato2011vc}  
      & Guarantees accuracy relying on the VC-dimension of queries \\
\cmidrule{2-4}
  & ML
    & KDE~\cite{gunopulos2005selectivity,gunopulos2000approximating,heimel2015self}
    & Applies kernel density estimation to selectivity estimation \\
  & & PGM~\cite{getoor-selectivity,graphical-sel1,graph-synopsis-selectivity}
    & Uses probabilistic graphical models for selectivity estimation \\
  & & Neural Net~\cite{cardinality-neural-net,kipf2018learned}
    & Trains a neural network for selectivity estimation \\
\midrule
\multirow{6}{*}{\parbox{16mm}{\textbf{Based on} \\ \textbf{Observed} \\ \textbf{Queries} \\ (Query-driven \\ Selectivity \\ Estimation) }}
  & \multirow{4}{*}{\parbox{20mm}{Error-feedback \\ Histograms \\ \emph{(fast but} \\ \emph{less accurate)}}}
    & ST-histogram~\cite{aboulnaga1999self} & 
      Refines the bucket frequencies based on the errors \\
  & & LEO~\cite{stillger2001leo}
      & Identifies incorrect statistics using observed queries \\
  & & STHoles~\cite{bruno2001stholes} 
      & Proposes a new buckets split mechanism; adopted by ISOMER \\
  & & SASH~\cite{lim2003sash} & Proposes a \emph{junction tree} model for finding the best set of histograms \\
  & & QueryModel~\cite{anagnostopoulos2015learning} & Avoids modeling the data distribution by using queries directly \\
\cmidrule{2-4}
  & \multirow{3}{*}{\parbox{20mm}{Max-Entropy \\ Histograms \\ \emph{(accurate but} \\ \emph{slow)} }} & ISOMER~\cite{srivastava2006isomer,markl2005consistently,markl2007consistent} 
    & Finds a maximum entropy distribution consistent with observed queries \\
& & Kaushik et al.~\cite{kaushik2009consistent} & Extends ISOMER for distinct values \\
& & R\'e et al.~\cite{re2012understanding,re2010understanding} & Seeks the max entropy distribution based on \emph{possible worlds} \\

\cmidrule{2-4}
  & \textbf{Mixture Model} \newline \em \textbf{(fast \& accurate)}
  & \textbf{\method (Ours)} & Employs a mixture model for selectivity estimation; develops an efficient training algorithm for the new model \\
\bottomrule
\end{tabular}
\end{table*}

\ph{Data Dimension}
Lastly, we studied the effect of data dimension
(i.e., number of columns) on error. 
Here, we generated the datasets using   multivariate normal distributions with different dimensions. 
For each dataset, 
we used three methods---AutoHist, AutoSample, and \method---to produce selectivity estimates for the predicates on all dimensions.
AutoHist used 1000 buckets, AutoSample used 1000 sampled rows, and \method used 1000 observed queries.

\cref{fig:exp:param:d} shows the result. The error of AutoHist increased quickly as the dimensions increased, which is a well-known problem of multidimensional histograms.
 In contrast, the errors of AutoSample and \method were not as sensitive to the number of dimensions of the data.
 The robustness of \method is due to the fact
that its internal model only depends on the intersection sizes of the query predicates (not on their dimensions).
Among these methods,  \method was the most accurate.
}



\ignore{
\subsection{\method's Optimization Efficiency}
\label{sec:exp:approx}
\label{sec:exp:opt}

To study the \method's optimization efficiency,
    we compared two approaches for solving the quadratic problem 
    defined in \cref{thm:quadratic}:
solving the original QP without 
    any modifications versus solving our modified 
version (\cref{prob:opt:alter}).
We used the  \texttt{cvxopt}
library for the former and used  \texttt{jblas} (a linear algebra library) for the latter. 
Both libraries use multiple cores for parallel processing. 

\cref{fig:exp:opt} shows
the time taken by   each optimization
approach. The second approach (\cref{prob:opt:alter})
was increasingly more 
efficient as the number of observed queries grew.
For example, it was 8.36$\times$ faster when the number of observed queries reached 1,000.
This is thanks to the modified problem 
having an analytical solution, while the original
problem required an iterative gradient descent 
solution.
}



\section{Connection: MSE and Entropy}
\label{sec:rel_to_max_entropy}

The max-entropy query-driven histograms optimize their parameters (i.e., bucket frequencies) by searching for the parameter values that maximize the entropy of the distribution $f(x)$. We show that this approach is approximately equivalent to \method's optimization objective, i.e., minimizing the mean squared error (MSE) of $f(x)$ from a uniform distribution. The entropy of the probability density function is defined as $\minus \int f(x) \, \log( f(x) ) \, dx$. Thus, maximizing the entropy is equivalent to minimizing $\int f(x) \, \log( f(x) ) \, dx$, which is related to minimizing MSE as follows:
\begin{align*}
  \argmin \int f(x)\,  \log( f(x) ) \; dx
  &\approx \argmin \int f(x)\, (f(x) - 1) \; dx \\
  &= \argmin \int (f(x))^2 \; dx
\end{align*}
since $\int f(x) \, dx = 1$ by definition. We used the first-order Taylor expansion to approximate $\log(x)$ with $x - 1$. Note that, when the constraint $\int f(x)\, dx = 1$ is considered, $f(x) = 1/|R_0|$ is the common solution to both the entropy maximization and minimizing MSE.

\ignore{
\subsection{Integration with Existing DBMS}
\label{sec:prelim:db}

\tofix{Frame this as a challenge}

In this paper, we study the query-driven selectivity estimation (\cref{prob:sec}) as a standalone problem, which is not tied to any specific DBMS.
However, any query-driven selectivity estimation technique  
(including ours)
    can be integrated into a DBMS using much of their 
    existing infrastructure.
Most DBMS systems contain the module that computes actual selectivities, the module that computes selectivity estimates, and the API to store metadata in its system catalog.
For example, Spark already collects actual selectivities
in the \texttt{FilterExec} class~\cite{spark_link1}.
Although Spark currently reports this selectivity  only at query time, it can be modified to also store the observed selectivities in its metastore (which is equivalent to a DB's system catalog).
The produced selectivity estimates can then be used in
the \texttt{FilterEstimation} class.
Prior work uses a similar integration strategy but for IBM DB2~\cite{stillger2001leo} and Microsoft SQL Server~\cite{agrawal2006autoadmin}.
}

\section{Related Work}
\label{sec:related}

There is   extensive  work on 
selectivity estimation due to its importance for query optimization.
In this section, we review both scan-based  (\cref{sec:related:scan}) and query-driven methods
(\cref{sec:related:queries}).
\method belongs to the latter category.
We have summarized the related work in \cref{tab:related}.


\subsection{Database Scan-based Estimation}
\label{sec:related:scan}
 
As explained in \cref{sec:intro},
we use the term \emph{scan-based methods} to refer to 
    techniques that directly inspect the data (or part of it) for collecting their statistics. 
These approaches differ from query-based methods which 
rely only on the actual selectivities of the observed queries.

\ph{Scan-based Histograms}
These approaches approximate the joint distribution by     
  periodically scanning the data. There has been
  much work on how to efficiently express the joint distribution of multidimensional data~\cite{cormode2011synopses,lynch1988selectivity,muralikrishna1988equi,van1993multiple,jagadish2001global,deshpande2001independence,thaper2002dynamic,ilyas2004cords,to2013entropy,lipton1990practical,haas1994relative,riondato2011vc,gunopulos2005selectivity,gunopulos2000approximating,heimel2015self,jestes2011building,lam2005dynamic,moerkotte2014exploiting,guha2002fast,istvan2014histograms,he2005k,karras2008lattice,chen2004multi}.
 There is also some work on  histograms 
 for special types of data, such as XML~\cite{wang2006decomposition,bhowmick2007efficient,aboulnaga2001estimating,wu2003using}, spatial data~\cite{zhang2004clustering,tao2003selectivity,mamoulis2001selectivity,wang2014selectivity,sun2006spatio,lin2003multiscale,zhang2002linear,koudas2000optimal,jagadish1998optimal,koloniari2005query,korn1999range,tang2014scalable,neumann2008smooth}, graph~\cite{feng2005dmt}, string~\cite{mazeika2007estimating,jagadish1999multi,jagadish2000one,jagadish1999substring}; or for privacy~\cite{hay2010boosting,kuo2018differentially,li2010optimizing}.

\ph{Sampling}
Sampling-based methods rely on a sample of data for estimating its joint distribution~\cite{lipton1990practical,haas1994relative,riondato2011vc}. However, drawing a new random sample 
    requires a table-scan or random retrieval of tuples,
        both of which are costly operations and hence,
        are only performed periodically.

\ph{Machine Learning Models}
\ignore{Here, we first describe \emph{Kernel density estimation} (KDE); then, describe its difference from mixture models (MM), which we employ for \method.}
 KDE is a technique that translates randomly sampled data points into a distribution~\cite{silverman2018density}.
 In the context of selectivity estimation, KDE has been
used as an alternative  to histograms~\cite{gunopulos2005selectivity,gunopulos2000approximating,heimel2015self}.
The similarity between KDE and mixture models (which we employ for \method) is that they both express a probability density function as a summation of some basis functions.
However, KDE and MM (mixture models) are fundamentally different. KDE relies on independent and identically distributed samples, and hence lends itself to  
scan-based selectivity estimation. 
In contrast, 
    MM does not require any
    sampling and can thus be used in
        query-driven selectivity estimation (where sampling is not practical).
Similarly,
probabilistic graphical models~\cite{getoor-selectivity,graphical-sel1,graph-synopsis-selectivity},
neural networks~\cite{cardinality-neural-net,kipf2018learned},
and tree-based ensembles~\cite{dutt2019selectivity}
have been used for selectivity estimation.
Unlike histograms, these approaches can  capture
column correlations more succinctly.
However, applicability of these models for query-driven selectivity estimation has not been explored and remains unclear.

More recently,
sketching~\cite{cai2019pessimistic} and
probe executions~\cite{trummer2019exact} have been proposed, 
which differ from ours in that they build their models directly using the data
 (not query results).
Similar to histograms, using the data requires
    either periodic updates or higher query processing overhead. \method avoids both of these shortcomings
with its query-driven MM.


\subsection{Query-driven Estimation}
 \label{sec:related:queries}

Query-driven techniques create their histogram buckets adaptively according to the queries they observe in the workload. These can be further 
categorized into two techniques based on  how they compute their bucket frequencies: error-feedback histograms and max-entropy histograms.

\ph{Error-feedback Histograms} Error-feedback histograms~\cite{aboulnaga1999self,bruno2001stholes,lim2003sash,anagnostopoulos2015learning,khachatryan2015improving,khachatryan2012sensitivity}  adjust bucket frequencies in consideration of the errors made by old bucket frequencies. 
 They differ in how they create histogram buckets according to the observed queries.
For example, STHoles~\cite{bruno2001stholes} splits existing buckets with the predicate range of the new query.
 SASH~\cite{lim2003sash} uses 
a space-efficient multidimensional histogram, called MHIST~\cite{deshpande2001independence}, but determines its bucket frequencies with an error-feedback mechanism.
 QueryModel~\cite{anagnostopoulos2015learning} 
 treats the observed queries themselves as conceptual histogram buckets and determines the distances among those buckets based on the similarities among the queries' predicates.


\ph{Max-Entropy Histograms}
Max-entropy histograms~\cite{srivastava2006isomer,markl2005consistently,markl2007consistent,kaushik2009consistent,re2012understanding} find a maximum entropy distribution  consistent with the observed queries. 
Unfortunately, these methods generally suffer from the exponential growth in their number of buckets as 
the number of observed queries grows (as discussed in \cref{sec:prelim}).
\method avoids this problem by relying on mixture models.

\ph{Fitting Parametric Functions}
Adaptive selectivity estimation~\cite{chen1994adaptive} fits a parametric function (e.g., linear, polynomial) 
to the observed queries. 
This approach is more applicable when we know the data distribution  \emph{a priori}, which is not assumed by \method.


\ph{Self-tuning Databases}
Query-driven histograms have also been studied in the context of self-tuning databases~\cite{DBLP:journals/pvldb/MarcusNMZAKPT19,DBLP:conf/cidr/KraskaABCKLMMN19,ma2018query}.
IBM's LEO~\cite{stillger2001leo} corrects  errors in any stage of query execution based on the observed queries.
Microsoft's AutoAdmin~\cite{agrawal2006autoadmin,chaudhuri2007self} focuses on automatic physical design, self-tuning histograms, and monitoring infrastructure. Part of this effort is ST-histogram~\cite{aboulnaga1999self} and STHoles~\cite{bruno2001stholes} (see \cref{tab:related}).
DBL~\cite{park2017database} and IDEA~\cite{galakatos2017revisiting} exploit the answers to past queries for more accurate approximate query processing.
 QueryBot 5000~\cite{ma2018query} forecasts the future queries,
whereas OtterTune~\cite{van2017automatic} 
and index~\cite{kraska2018case}
use machine learning for automatic physical design and 
building secondary indices, respectively.

\section{Conclusion and Future Work}

The prohibitive cost of 
query-driven selectivity estimation techniques 
    has greatly limited their adoption by DBMS vendors, which for the most part still rely on scan-based histograms and samples that are periodically updated
    and are otherwise stale.
In this paper, we proposed a new framework, called \emph{selectivity learning}
or \method, which learns from every query to continuously refine its internal model of the underlying data, and therefore produce increasingly more accurate selectivity estimates over time. 
\method differs from previous query-driven selectivity estimation techniques by (i) not using histograms
and (ii) enabling extremely fast refinements using its mixture model.
We formally showed that
 the training cost of our mixture model   can be reduced from exponential   to only quadratic complexity (\cref{thm:quadratic}).
%
%
%


\ph{Supporting Complex Joins}
When modeling the selectivity of join queries,
even state-of-the-art modes ~\cite{kipf2018learned,dutt2019selectivity,deep-likelihood}
take a relatively simple approach:
conceptually prejoining corresponding tables,
and constructing a joint probability distribution over each join pattern.
We plan to similarly extend our current formulation of \method 
 to model the  selectivity of general joins.
\ignore{however,
since there may be $O(2^N)$ join patterns given $N$ tables,
this approach soon becomes infeasible.
    Instead, we think we need a combination of sampling
    and models. Unlike many models, sampling is 
    more robust to curse of dimensionality (i.e., 
    we can more easily handle many tables with a large number
    of columns); however,
    unlike models, sampling is less efficient in capturing
    the data distribution (as shown in our experiments).
    By combining these two, however, we will be able
    to achieve accurate estimation even for high-dimensional data.
}



\section{Acknowledgement}

This material is based upon work supported by the National Science Foundation under Grant No.~1629397 and the Michigan Institute for Data Science (MIDAS) PODS.
Any opinions, findings, and conclusions or recommendations expressed in this material are those of the author(s) and do not necessarily reflect the views of the National Science Foundation.

\bibliographystyle{ACM-Reference-Format}
\bibliography{biblio/selectivity,biblio/mozafari,biblio/approximate,biblio/predictability}  


%
%

\clearpage

\appendix

\section{Analysis of Iterative Scaling}
\label{sec:appendix}

\newcommand{\lamvec}{\bm{\lambda}}

The previous work uses an algorithm called \emph{iterative scaling} to optimize the bucket frequencies. In this section, we analyze why using the approach is non-trivial when some buckets are merged or pruned.
Specifically, we show that the approach becomes non-trivial if a histogram bucket may partially overlap with a predicate range, rather than being completely within or outside of the predicate range.

\ph{Selectivity Estimation with Histograms}
First, we describe how histograms can be used for selectivity estimation. This description is needed to present iterative scaling itself.
Let $G_z$ for $z = 1, \ldots, m$ denote the boundary of $z$-th bucket. 
$w_z$ is the frequency of the $z$-th bucket.
Then, the histogram approximates the distribution of data as follows:
\[
f(x) = \frac{w_j}{|G_j|} \qquad \text{for } G_j \text{ such that } x \in G_j
\]
For a query's predicate $P_i$, its selectivity can be computed as follows:
\[
s_i = \sum_{j=1}^m \frac{|B_i \cap G_j|}{|G_j|} \, \ww = \ss
\]
One can observe that the above expression is akin to selectivity estimation formula with a mixture model (\cref{sec:method:est}), which is natural since mixture models can be regarded as a generalization of histograms.

\ph{Optimization with Maximum Entropy Principle}
To optimize the bucket frequencies, the previous work employs the maximum entropy principle.
When the maximum entropy principle is used, one can optimize the bucket frequencies by solving the following problem.

\begin{problem}[Training with Max Entropy Principle]
\begin{align*}
\argmin_{\ww} \quad
&\int \; f(x) \, \log(f(x)) \; dx \\
\text{such that} \quad 
& A \, \ww = \ss
\end{align*}
where $A$ is a $n$-by-$m$ matrix; its $(i,j)$-th entry is defined as
\[
(A)_{i,j} = \frac{|B_i \cap G_j|}{|G_j|}.
\]
$\ss$ is a size-$n$ column vector; its $i$-th entry is the observed selectivity of the $i$-th query.
\end{problem}

If a histogram bucket is completely included within a predicate range, $A_{i,j}$ takes 1; if a histogram bucket is completely outside a predicate range, $A_{i,j}$ takes 0. If a histogram bucket partially overlaps with a predicate range, $A_{i,j}$ takes a value between 0 and 1.

For histograms, the integral in the above problem can be directly simplified as follows:
\begin{align*}
\int \; f(x) \, \log(f(x)) \; dx 
&= \sum_{i=1}^m \, 
     \int_{x \in G_i} \; \frac{w_i}{|G_i|} \log\left( \frac{w_i}{|G_i|} \right)
   \; dx  \\
&= \sum_{i=1}^m \, w_i \, \log\left( \frac{w_i}{|G_i|} \right)
\end{align*}

\ph{Iterative Scaling}
Iterative Scaling solves the above problem by updating model parameters in a sequential order; that is, it updates $w_1$ as using fixed values for other parameters, it updates $w_2$ as using fixed values for other parameters, and so on. This iteration (i.e., updating all $w_1$ through $w_m$) continues until those parameter values converge. In this process, the important part is the formula for the updates.

To derive this update rule, the previous work uses the Lagrangian method, as follows. In the following derivation, we suppose a slightly more general setting; that is, we allow possible partial overlaps. We first define $L$:
\[
L(\ww, \lamvec) = \sum_{j=1}^m w_j 
  \log\left( \frac{w_j}{|G_j|} \right)
  - \lamvec^\top (A \ww - \ss)
\]
where $\lamvec$ is a size-$m$ column vector containing $m$ Lagrangian multipliers, i.e.,  $\lamvec = (\lambda_1, \ldots, \lambda_m)^\top$.

At optimal solutions, the derivative of $L$ with respect to $\ww$ and $\lambda$ must be zero. 
Let the column vectors of $A$ be denoted by $a_1, \ldots, a_m$; that is, $A = [a_1, \ldots, a_m]$.
Then,
\begin{align*}
\frac{\partial L}{\partial w_j}
= \log(w_j) + 1 - \log(|G_j|) - a_j^\top \lamvec = 0
\end{align*}
Let $z_i = \exp(\lambda_i)$. Then,
\begin{align}
&\log \left( \frac{w_j}{|G_j|} \right) = a_j^\top \lamvec - 1  
  \qquad \text{for } j = 1, \ldots, m \nonumber \\
\Rightarrow \quad & \frac{w_j}{|G_j|} = \frac{1}{e} 
  \exp \left( a_j^\top \lamvec \right) 
    \qquad \text{for } j = 1, \ldots, m \nonumber \\
\Rightarrow \quad & w_j = \frac{|G_j|}{e} 
  \prod_{i=1}^n \exp \left( A_{i,j} \, \lambda_i \right)
    \qquad \text{for } j = 1, \ldots, m \nonumber \\
\Rightarrow \quad & w_j = \frac{|G_j|}{e} 
  \prod_{i=1}^n z_i^{A_{i,j}}
    \qquad \text{for } j = 1, \ldots, m \label{eq:is:freq}
\end{align}
From the constraint that $\sum_{j=1}^m A_{i,j} w_j = s_i$ for $i = 1, \ldots, n$,
\begin{align}
\sum_{j=1}^m A_{i,j} \frac{|G_j|}{e}
  \prod_{i=1}^n z_i^{A_{i,j}} = s_i
  \label{eq:is:pre}
\end{align}
\textbf{Let's assume} that $A_{i,j}$ always takes either 0 or 1 (the condition used in the previous work); then, the above expression can be simplified to produce an analytic update rule. Let $C_i$ be an index set such that $C_i = \{ k \mid A_{i,k} = 1, k = 1, \ldots, m \}$. Also, let $D_{j \textbackslash i} = \{ k \mid A_{k,j} = 1, k = 1, \ldots, n, k \ne i \}$. Then,
\begin{align}
&\sum_{j=1}^m A_{i,j} \frac{|G_j|}{e}
  \prod_{i=1}^n z_i^{A_{i,j}} = s_i \nonumber \\
\Rightarrow \quad &
  \sum_{j \in C_i} \frac{|G_j|}{e} \; z_i \prod_{k \in D_{i \textbackslash j}} z_k = s_i \nonumber \\
\Rightarrow \quad &
  z_i = \frac{s_i}{\sum_{j \in C_i} |G_j| / e \prod_{k \in D_{i \textbackslash j}} z_k}
  \label{eq:is:update}
\end{align}
Using the above equation, the previous work continues to update $z_i$ for $i = 1, \ldots, n$ until convergence. Once those values are obtained, the bucket frequencies can be obtained by \cref{eq:is:freq}.

\textbf{However}, without the assumption that $A_{i,j}$ always takes either 0 or 1,
obtaining the update equation (\cref{eq:is:update}) from \cref{eq:is:pre} is non-trivial.

\end{document}